\documentclass[12pt]{iopart}
\usepackage{iopams,mathptm,times,graphicx}

\newcounter{theorem}
\renewcommand\thetheorem{\arabic{section}.\arabic{theorem}}

\newenvironment{lemma}{\par\medskip\noindent\begingroup{\bf Lemma
             \stepcounter{theorem}\thetheorem.}\ \itshape
             \def\@currentlabel{\thetheorem}}{\endgroup\par\medskip}
\newenvironment{theorem}{\par\medskip\noindent\begingroup{\bf Theorem
             \stepcounter{theorem}\thetheorem.}\ \itshape
             \def\@currentlabel{\thetheorem}}{\endgroup\par\medskip}

\newenvironment{proof}{\par\noindent{\bf Proof.} }{\proofbox\par\medskip}

\def\proofbox{\hfill{\ensuremath\Box}}

\usepackage{color}

\begin{document}
\title[]{General soliton solution to a nonlocal nonlinear Schr\"odinger
equation with zero and nonzero boundary conditions}
\author{Bao-Feng Feng$^1$, Xu-Dan Luo$^2$, Mark J. Ablowitz$^3$, Ziad H. Musslimani$^4$}
\address{$^1$~School of Mathematical and Statistical Sciences,
The University of Texas Rio Grande Valley, Edinburg, TX 78539}
\address{$^2$~Department of Mathematics, State University of New York at Buffalo, Buffalo, NY 14260}
\address{$^3$~Department of Applied Mathematics, University of Colorado, Boulder, Colorado 80309-0526}
\address{$^4$~Department of Mathematics, Florida State University, Tallahasse, Florida 32306-4510}
\date{\today}

\begin{abstract}
{General soliton solutions} to a nonlocal nonlinear Schr\"{o}dinger (NLS) equation with PT-symmetry for both zero and nonzero boundary conditions {are considered} via the combination of Hirota's bilinear method and the Kadomtsev-Petviashvili (KP) hierarchy reduction method. First, general $N$-soliton solutions with zero boundary conditions are constructed. Starting from the tau functions of the two-component KP hierarchy, {it is shown that} they can be expressed in terms of either Gramian or double Wronskian determinants. {On the contrary}, from the tau functions of  single component KP hierarchy, general soliton solutions to the nonlocal NLS equation with nonzero boundary conditions are obtained. All possible soliton solutions to nonlocal NLS with Parity (PT)-symmetry for both zero and nonzero boundary conditions are found in the present paper.
\end{abstract}




\section{Introduction}
The study of the nonlinear Schr\"odinger (NLS) equation {and its multi-component generalization}
lies at the forefront of {research in} applied mathematics and mathematical physics. {This is the case} since {they have} been recognized as generic models for describing the evolution of slowly varying wave packets in general nonlinear wave system \cite{Benney1967,Ablowitzbook,APT}.
In 1972 Zakharov and Shabat found that the NLS equation possesses a Lax
pair and can be solved/linearized by inverse scattering \cite{ZakharovShabat}.
Motivated by these results in 1973 Ablowitz, Kaup, Newell and Segur
\cite{AKNS} generalized the linear operators used by Zakharov and Shabat; they
showed that the NLS, sine-Gordon\cite{AKNSsG}, modified-KdV, KdV all could be
solved/linearized by inverse scattering. Soon afterwards in 1974, AKNS
\cite{AKNSIST} developed a general framework to find integrable systems solvable by
what they termed the Inverse Scattering Transform (IST). The method
was associated with classes of equations (later called recursion operators)
and was applied to find solutions to the initial value problem with rapidly
decaying data on the line.

The NLS and coupled NLS equations arise in a variety of physical contexts such as nonlinear optics \cite%
{HasegawaTappert1,HasegawaTappert2,Agrawalbook,HasegawaKodama,AgrawalKivsharbook}%
, Bose-Einstein condensates \cite{BECReview}, water waves {\cite{Zakharov1968, BenneyRoskes1969}}
and plasma physics \cite{ZakharovPlasma}{, amongst many others.}

The idea in \cite{AKNSIST} was to consider the scattering problem
\begin{equation}  \label{AKNS_space}
v _{x}=Xv = \left(
\begin{array}{cc}
-\mathrm{i}k & q(x,t) \\
r(x,t) & \mathrm{i}k%
\end{array}%
\right)\,,
\end{equation}
\begin{equation}
\label{AKNS_time}
v_{t}=Tv=\left(
\begin{array}{cc}
2\mathrm{i}k^{2}+\mathrm{i}qr & -2kq-\mathrm{i}q_{x} \\
-2kr+\mathrm{i}r_{x} & -2\mathrm{i}k^{2}-\mathrm{i}qr%
\end{array}%
\right)\,.
\end{equation}%
The compatibility condition of the linear spectral problem (\ref{AKNS_space})--(\ref{AKNS_time})
leads to
\begin{equation}
\left\{
\begin{array}{l}
\displaystyle\mathrm{i}q_{t}=q_{xx}-2rq^{2}\,, \\
\displaystyle -\mathrm{i}r_{t}=r_{xx}-2qr^{2}\,.%
\end{array}%
\right.  \label{AKNS_bilinear}
\end{equation}%
Under the standard symmetry reduction
\begin{equation}
r(x,t)=\sigma q^{\ast }(x,t)\,, \quad \sigma = \pm 1\,,
\end{equation}
the system (\ref{AKNS_bilinear})  gives the classical NLS
equation
\begin{equation}  \label{NLS}
\mathrm{i}q_{t}(x,t)=q_{xx}(x,t)-2\sigma q^{\ast }(x,t)\,q^{2}(x,t)\,.
\end{equation}
Recently, an integrable nonlocal NLS equation
\begin{equation}  \label{NLS_PT}
\mathrm{i}q_{t}(x,t)=q_{xx}(x,t)-2\sigma q^{\ast }(-x,t)\,q^{2}(x,t)
\end{equation}
was found in \cite{AblowitzMusslimani} under a new symmetry reduction
\begin{equation}  \label{PT-reduction}
r(x,t)=\sigma q^{\ast }(-x,t)\,,\quad \sigma = \pm 1\,.
\end{equation}
This nonlocal equation is shown to be related to an unconventional magnetic system \cite{ApplicNlclNLS16}. It is parity-time (PT) symmetric, i.e., it is invariant under
the joint transformations of $x \to -x$, $t \to t$ and
complex conjugation, thus is related to a hot research area of contemporary
physics \cite{Bender98,KonotopYangRMP}.
Due to this potential application, it has been shown that different symmetry reductions from the linear problem of general AKNS hierarchy and other integrable hierarchies can {also} lead to various types of new nonlocal equations, which have been extensively studied \cite{AblowitzMusslimani3}--\cite{JYangNonlocalNLS}.  Typical examples are the
reverse space-time nonlocal NLS equation and the reverse time nonlocal NLS
equation, the complex/real space-time Sine-Gordon equation \cite%
{AblowitzMusslimani3,AblowitzMusslimani4,AblowitzLuoMusslimani,AblowitzFengLuoMusslimani}, the complex/real reverse
space-time mKdV equation \cite{AblowitzMusslimani3,AblowitzMusslimani4}, the
nonlocal derivative NLS equation \cite{ZXZhou}, and the multi-dimensional
nonlocal Davey-Stewartson equation \cite{AblowitzMusslimani4,Fokas2015}. On
the other hand, corresponding to the classical semi-discrete NLS equation \cite{AblowitzLadik,AL2}, the nonlocal NLS
equation also admits its semi-discrete version \cite{AblowitzMusslimani2} and
{multi-component generalizations} \cite{AblowitzMusslimani4,ZenyanAML15}.

It has been known that that the inverse scattering transform (IST) has been successfully developed for the NLS, coupled NLS, as well as their semi-discrete analogues \cite{Prinari06}--\cite{BPrinari}.
Moreover, the IST to solve the initial-value problem
with vanishing boundary condition for the nonlocal NLS equation has been
developed in \cite{AblowitzMusslimani,AblowitzMusslimani3,AblowitzMusslimani4}.
For pure soliton solutions corresponding to reflectionless potentials, one obtains
one- and two-soliton solutions; these solutions can have singularities. 
In particular, the {decaying} one-soliton solution is found to be \cite{AblowitzMusslimani}
\begin{equation}
q=\frac{2(\eta +\bar{\eta})e^{-2\bar{\eta}x-4\mathrm{i}{\bar{\eta}}^{2}t+%
\mathrm{i}\theta }}{1-e^{-2(\eta +\bar{\eta})x+4\mathrm{i}(\eta ^{2}-{\bar{%
\eta}}^{2})t+\mathrm{i}(\theta +\bar{\theta})}}\,.  \label{one-bright}
\end{equation}%
{This solution is singular at $x=0$ when $\eta \neq \bar{\eta}$ and  $4(\eta^2-\bar{\eta}^2)t + (\theta +\bar{\theta})= 2n \pi, n \in \mathbb{Z}$. When $\eta = \bar{\eta}, \theta +\bar{\theta} \neq 2n \pi, n \in \mathbb{Z}$ the solution is not singular. }
It is also noted that the same solution has been found in \cite{ZhuMa2,DajunAML17} via Darboux transformations and AKNS reductions, respectively.

The inverse scattering transform (IST) to solve the initial-value problem
with nonzero boundary conditions has also been developed in \cite%
{AblowitzLuoMusslimani}. To be specific, under the following boundary
conditions
\begin{equation}
q(x,t)\rightarrow q_{\pm }(t)=q_{0}e^{i(\alpha t+\theta _{\pm })},~\ as\
x\rightarrow \pm \infty ,  \label{E:NZBC}
\end{equation}%
where $q_{0}>0$, $0\leq \theta _{\pm }<2\pi $, the following cases are {studied.}
\begin{itemize}
  \item $\sigma =-1$, {$\Delta=\theta _{+}-\theta _{-}=\pi $}: {a} one-soliton solution
is found{:}
\begin{equation}
q=q_{0}e^{\mathrm{i(}{2q}_{0}^{2}t+\theta _{+}-\pi )}\tanh [q_{0}x-\mathrm{i}%
\theta _{\ast }]\,
\end{equation}
where $\theta _{\ast }=\frac{1}{2}(\theta _{+}+\theta _{1}+\pi ),$ {where $\theta _{1}$ is related
to the associated scattering data (normalization constant).} This
solution is stationary in space and oscillating in time.
  \item $\sigma =-1$, {$\Delta=0$}: in this case, there is no 'proper
exponentially decaying' pure one soliton solution. The simplest decaying
pure reflectioness potential generates a two-soliton solution.
  \item $\sigma =1$, {$\Delta=0 $}: 
 {in this case only an even number of soliton solutions arise; they are related to an even number ($2N$) of eigenvalues.}
  \item $\sigma =1$, {$\Delta=\pi$}: in this case, there is no eigenvalues/solitons.
\end{itemize}
Regardless of the success of the IST for the nonlocal NLS equation, it is still not easy to come up with its general soliton solution in closed form by this method. On the other hand, Hirota's bilinear method is a direct and powerful method in finding multi-soliton solutions to soliton equations \cite{HirotaBook}. In the early 1980's, suggested by Hirota's direct method, Sato discovered that the soliton solutions to the Kadomtsev-Petviashvili (KP) equation and its hierarchy (named tau-functions afterwards) satisfy the so-called Pl\"ucker relations, and found that the totally of solutions to the KP hierarchy forms an infinite dimensional Grassmann manifold \cite{Sato81}. Later on, Ueno and Takasaki extended Sato's discovery to include two-dimensional Toda lattice (2DTL) hierarchy \cite{UenoTakasaki}; Date, Jimbo and Miwa further developed Sato's idea with transformation groups theory, and presented Lie algebraic classification of soliton equations \cite{JM83,DKJM83,MJD}.

The KP-hierarchy reduction method was  {first} developed by the Kyoto school
\cite{JM}, and later was used to obtain soliton solutions in the NLS
equation, the modified KdV equation, the Davey-Stewartson equation and the
coupled higher order NLS equations \cite{OhtaRIMS1989,GHNO}. Recently, this
method has been applied to derive dark-dark soliton solution in two-coupled
NLS equation of the mixed type \cite{OhtaJianke}. {Unlike the} 
the inverse scattering transform method \cite{Prinari06}, and Hirota's
bilinear method \cite{HirotaBook}, the KP-hierarchy reduction method starts
with the general KP hierarchy including the two-dimensional Toda-hierarchy
\cite{MiyakeOhtaGram} and derives the general soliton solution in either
determinant or pfaffian form reduced directly the tau functions of the KP
hierarchy.
{In spite of many successes of the KP-hierarchy reduction method, constructing general soliton solutions by this method to the nonlocal NLS equation remains an unsolved and challenging problem, especially for the nonzero boundary conditions. The difficulty lies in the simultaneous realization of both the nonlocal and complex conjugate reductions. Therefore, it motivates the present work, which intends to construct general soliton solutions in determinant forms to the nonlocal NLS equation with both the zero and nonzero boundary conditions.}

 In the present paper, we find solitons that are both decaying and have nonzero boundary values.
For the decaying boundary condition, we construct the general $N$-soliton solutions expressed in both the Gram-type and the double-Wronskin determinants. The explicit one-soliton solution agrees with the solution found by the inverse scattering transform \cite{AblowitzMusslimani,JYangNonlocalNLS}, which is also listed in (\ref{one-bright}). For the nonzero boundary condition, the general soliton solutions are constructed for all cases except for the case: $\sigma =1$, $\Delta=\pi$ where there is no soliton solution. To be specific, for the case of $\sigma =-1$, $\Delta=\pi$, the solution (\ref{Sig-1Detpi}) is the same as above; for the case of $\sigma =1$, $\Delta=0$, it is true that only the determinant soliton solution $2N \times 2N$ can be constructed, which is consistent with the fact that only the even number of soliton solutions with an even number of eigenvalues exist. The simplest solution based on the tau-functions found here corresponds to (5.90) in \cite{AblowitzLuoMusslimani} but with a simpler form. In the last, for the case of $\sigma =-1$, $\Delta=0$, the soliton solution (\ref{Sig01Det0}) found in this paper is the same as the one found \cite{AblowitzLuoMusslimani}.

The outline of the present paper is organized as follows. In section 2, we
construct the bright soliton solution to the nonlocal NLS equation with {a} zero boundary condition.
Both the Gram-type and double Wronski-type solution {start} 
from the tau functions in Gram-type or double Wronski-type of two-component KP hierarchy.
In section 3, general multi-soliton solutions for the nonlocal NLS equation of $\sigma=1$  and of $\sigma=-1$ are constructed, which recover all possible solutions found in \cite{AblowitzLuoMusslimani}.
The paper {concludes with} 
comments and remarks in section 4.
\section{General soliton solution with zero boundary condition}
In this section, we consider general soliton solutions to the nonlocal NLS equation (\ref{NLS_PT}).
under the boundary condition $q(x,t) \to 0$ as $x \to \pm \infty$. This type of soliton solution is usually called bright soliton solution. We first give the bilinear forms for  the nonlocal NLS equation (\ref{NLS_PT}).
By introducing the dependent variable transformation
\begin{equation}
q(x,t)=\frac{g(x,t)}{f(x,t)},
\end{equation}
the nonlocal NLS equation (\ref{NLS_PT})  is converted into the following
bilinear equations
\begin{equation}
\left\{
\begin{array}{l}
\displaystyle(\mathrm{i}D_{t}-D_{x}^{2})g(x,t)\cdot f(x,t)=0,\quad \\
\displaystyle D_{x}^{2}f(x,t)\cdot f(x,t)+2\sigma g(x,t)g^{\ast }(-x,t)=0\,,%
\end{array}%
\right.  \label{NLSPT_bilinear}
\end{equation}%
where the Hirota's bilinear operator is defined as
\[
D_{x}^{n}D_{t}^{l}f(x,t)\cdot g(x,t)=\left( \frac{\partial }{\partial x}-%
\frac{\partial }{\partial x^{\prime }}\right) ^{n}\left( \frac{\partial }{%
\partial t}-\frac{\partial }{\partial t^{\prime }}\right)
^{l}f(x,t)g(x^{\prime },t^{\prime })\bigg|_{x=x^{\prime },t=t^{\prime }}.
\]%
Similar to the classical NLS equation, the bright soliton solutions to the nonlocal NLS can be obtained from the tau functions of two-component KP hierarchy. It is known that the tau functions have two
different forms: one is the Gram-type, the other is the Wronski-type. They are basically equivalent, however, each form has advantage and disadvantage in view of different symmetries. In what follows, we will derive two types of soliton solutions via KP hierarchy reduction method.
\subsection{General bright soliton solution expressed by Gramian determinant}
In this subsection, we will construct soliton solution in terms of Gram-type determinant.  To this end, let us start with a Gram-type determinant
expression of the tau functions for two-component KP hierarchy,
\begin{equation}
\tau _{0}=\left\vert A\right\vert\,,
\end{equation}

\begin{equation}
\tau _{1}=\left\vert
\begin{array}{cc}
A & \Phi ^{T} \\
-\bar{\Psi} & 0%
\end{array}%
\right\vert, \quad \tau _{-1}=\left\vert
\begin{array}{cc}
A & \Psi ^{T} \\
-\bar{\Phi} & 0%
\end{array}%
\right\vert
\end{equation}
where the elements of matrix $A$ are
\[
a_{ij}=\frac{1}{p_{i}+\bar{p}_{j}}e^{\xi _{i}+\bar{\xi}_{j}}+\frac{1}{q_{i}+%
\bar{q}_{j}}e^{\eta _{i}+\bar{\eta }_{j}}
\]
with
\[
\xi _{i}=p_{i}x_{1}+p_{i}^{2}x_{2}+\xi _{i0},\quad \bar{\xi}_{j}=\bar{p}%
_{j}x_{1}-\bar{p}_{j}^{2}x_{2}+\bar{\xi}_{j0},
\]%

\[
\eta _{i}=q_{i}y_{1}+\eta _{i0},\quad \bar{\eta }_{j}=\bar{q}%
_{j}y_{1}+\bar{\eta }_{j0},
\]
and $\Phi$, $\bar{\Phi}$, $\Psi$, $\bar{\Psi}$ are row vectors defined by
\[
\Phi =\left( e^{\xi _{1}},\cdots ,e^{\xi _{N}}\right) \,,\quad  \Psi =\left(
e^{\eta _{1}},\cdots ,e^{\eta _{N}}\right) \,,
\]%
\[
\bar{\Phi}=\left( e^{\bar{\xi}_{1}},\cdots ,e^{\bar{\xi}_{N}}\right) \,,\quad
\bar{\Psi}=\left( e^{\bar{\eta}_{1}},\cdots ,e^{\bar{\eta}_{N}}\right) \,.
\]
It can be shown below that the tau functions given above satisfy the following two bilinear
equations
\begin{equation}
\label{bilinear_bright}
\left\{
\begin{array}{l}
\displaystyle(D_{x_{2}}-D_{x_{1}}^{2})\tau _{1}\cdot \tau _{0}=0\,, \\[5pt]
\displaystyle D_{x_{1}}D_{y_{1}}\tau _{0}\cdot \tau _{0}=-2\tau _{1}\tau
_{-1}\,.%
\end{array}%
\right.
\end{equation}
The proof is given as follows.  Since%
\[
\partial _{x_{1}}\tau _{0}=\left\vert
\begin{array}{cc}
A & \Phi ^{T} \\
{-}\bar{{\Phi }} & 0%
\end{array}%
\right\vert, \quad \partial _{x_{1}}\partial _{y_{1}}\tau _{0}=\left\vert
\begin{array}{ccc}
A & \Phi ^{T} & \Psi ^{T} \\
{-}\bar{{\Phi }} & 0 & 0 \\
{-}\bar{{\Psi }} & 0 & 0%
\end{array}%
\right\vert\,,
\]
then based on the Jacobi identify for the determinant, we have
\[
(2\partial _{x_{1}}\partial _{y_{1}})\tau _{0}\times \tau _{0}=2(\partial
_{x_{1}}\tau _{0})\times (\partial _{x_{1}}\tau _{y})-2\tau _{1} \times \tau
_{-1}\,,
\]
which is exactly the second bilinear equation.

Moreover, we can easily verify the following relations
\[
\partial _{x_{2}}\tau _{0}=\left\vert
\begin{array}{cc}
A & \Phi _{x_{1}}^{T} \\
{-}\bar{{\Phi }} & 0%
\end{array}%
\right\vert -\left\vert
\begin{array}{cc}
A & \Phi _{x_{1}}^{T} \\
{-}\bar{\Phi }_{x_{1}} & 0%
\end{array}%
\right\vert\,,
\]

\[
\partial _{x_{1}}^{2}\tau _{0}=\left\vert
\begin{array}{cc}
A & \Phi _{x_{1}}^{T} \\
{-}\bar{{\Phi }} & 0%
\end{array}%
\right\vert +\left\vert
\begin{array}{cc}
A & \Phi _{x_{1}}^{T} \\
{-}\bar{\Phi }_{x_{1}} & 0%
\end{array}%
\right\vert\,, \quad \partial _{x_{1}}\tau _{1}=\left\vert
\begin{array}{cc}
A & \Phi _{x_{1}}^{T} \\
{-}\bar{{\Psi }} & 0%
\end{array}%
\right\vert\,,
\]
and
\[\partial _{x_{1}}^{2}\tau _{1}=\left\vert
\begin{array}{ccc}
A & \Phi ^{T} & \Phi _{x_{1}}^{T} \\
{-}\bar{{\Phi }} & 0 & 0 \\
{-}\bar{{\Psi }} & 0 & 0%
\end{array}%
\right\vert +\left\vert
\begin{array}{cc}
A & \Phi _{x_{1}x_{1}}^{T} \\
{-}\bar{{\Psi }} & 0%
\end{array}%
\right\vert\,,
\]

\[
\partial _{x_{2}}\tau _{1}=-\left\vert
\begin{array}{ccc}
A & \Phi ^{T} & \Phi _{x_{1}}^{T} \\
{-}\bar{{\Phi }} & 0 & 0 \\
{-}\bar{{\Psi }} & 0 & 0%
\end{array}%
\right\vert +\left\vert
\begin{array}{cc}
A & \Phi _{x_{1}x_{1}}^{T} \\
{-}\bar{{\Psi }} & 0%
\end{array}%
\right\vert\,,
\]
where $\Phi _{x_{1}}$ and $\Phi _{x_{1}x_{1}}$ represent the first and second derivatives of $\Phi$ respect to
$x_1$, respectively. Therefore we have
\[
(\partial _{x_{2}}-\partial _{x_{1}}^{2})\tau _{0}=2\left\vert
\begin{array}{cc}
A & \Phi _{x_{1}}^{T} \\
{-}\bar{{\Psi }} & 0%
\end{array}%
\right\vert\,,
\]

\[
(\partial _{x_{2}}-\partial _{x_{1}}^{2})\tau _{1}=-2\left\vert
\begin{array}{ccc}
A & \Phi ^{T} & \Phi _{x_{1}}^{T} \\
{-}\bar{{\Phi }} & 0 & 0 \\
{-}\bar{{\Psi }} & 0 & 0%
\end{array}%
\right\vert\,.
\]
From the Jacobi identity for the determinant, we have%
\[
\left((\partial _{x_{2}}-\partial _{x_{1}}^{2}) \tau _{1} \right)\times \tau _{0}=\left((\partial
_{x_{2}}-\partial _{x_{1}}^{2})\tau _{0}\right) \times \tau _{1}-2(\partial
_{x_{1}}\tau _{1})\times (\partial _{x_{1}}\tau _{f
0})
\]
which gives the first bilinear equation. The proof is done.

Next, we show the reduction processes. By row
operations $\tau _{0}$ can be rewritten as
\[
\tau _{0}=\prod_{j=1}^{N}e^{\xi _{j}+\bar{\xi}_{j}}\left\vert
a_{ij}^{{\prime }}\right\vert_{1\le i,j \le N} =\prod_{j=1}^{N}e^{\xi _{j}+\bar{\xi}%
_{j}}\left\vert \frac{1}{p_{i}+\bar{p}_{j}}+\frac{1}{q_{i}+\bar{q}_{j}}%
e^{(\eta _{i}-\xi _{i})+(\bar{\eta}_{j}-\bar{\xi}_{j})}\right\vert \,
\]%
where
\[
\eta _{i}-\xi _{i}=q_{i}y_{1}-p_{i}x_{1}+\cdots \,,
\quad \bar{\eta}_{j}-\bar{\xi}_{j}=\bar{q}_{j}y_{1}-\bar{p}_{j}x_{1}+\cdots \,.
\]
Thus if we impose constraints on parameters
\begin{equation}
q_{j}=p_{j}\,,\quad \bar{q}_{j}=\bar{{p}}_{j}\,,\quad j=1,\cdots ,N,
\end{equation}%
then the following relation holds
\[
(\partial _{x_{1}}+{\partial _{y_{1}})}a_{ij}^{\prime}=0
\]%
which implies
\begin{equation}
\label{dimension_redcution}
\partial _{x_{1}}\tau _{0}=-\partial _{y_{1}}\tau _{0}.
\end{equation}
Consequently, the second bilinear equation in (\ref{AKNS_bilinear}) becomes
\begin{equation}
\label{AKNS_bilinearb}
D_{x_{1}}^{2}\tau _{0}\cdot \tau _{0}=2\tau _{1}\tau _{-1}\,.
\end{equation}
Due to the dimension reduction (\ref{dimension_redcution}),
$y_{1}$ becomes a dummy variable, which can be treated as zero.
We then have the following Lemma.
\begin{lemma}
Assume $x_{1}=x$ is real and $x_{2}=-\mathrm{i}t$ is
purely imaginary variable. If $p_{j}$ and $\bar{p}_{j}$ ($j=1,\cdots, N$) are real, ${\eta }_{j0}$, $\bar{\eta}_{j0}$ are puly imaginary, or the subset of ($p_{j}$, ${\eta }_{j0}$) or ($\bar{p}_{j}$, $\bar{\eta}_{j0}$) occurs in pair such that $p_{k}=p^*_{k'}$, ${\eta }_{k0}={\eta}^*_{k'0}$ or $\bar{p}_{k}=\bar{p}^*_{k'}$, $\bar{\eta}_{k0}=\bar{\eta}^*_{k'0}$, then
\begin{equation}\label{bright_symmetry}
  \tau _{0}(x,t)=C \tau^{\ast} _{0}(-x,t)\, \quad \tau_{-1}(x,t)=C \tau^{\ast} _{1}(-x,t)\,,
\end{equation}
where $C=C_1 C_2$, $C_1=\prod_{j=1}^{N} e^{({\xi _{j}+\bar{\xi}_{j}})}$,
$C_2=\prod_{j=1}^{N} e^{(\eta _{j}+\bar{\eta }_{j})}$.
\end{lemma}
\begin{proof}  First, we prove the case when $p_{i}$ and $\bar{p}_{i}$ are
all real, ${\eta }_{i0}$, $\bar{\eta}_{i0}$ are all purely imaginary. Under
this case,
\[
\xi _{i}=p_{i}x-{\mathrm{i}}p_{i}^{2}t\,,
\quad \eta _{i}=\mathrm{i}\theta_{i}\,,
\]
\[
\bar{\xi}_{j}=\bar{p}_{j}x+\mathrm{i}
\bar{p}_{j}^{2}t\,, \quad \bar{\eta }_{j}= \mathrm{i}\bar{\theta}_{j}\,,
\]
it then follows
\[
\xi _{i}^{\ast}(-x,t)=-\xi _{i}(x,t)\,, \quad  \bar{\xi}_{i}^{\ast }(-x,t)=-\bar{\xi}_{i}(x,t)\,.
\]
We note that $\tau _{0}(x,t)$ can be rewritten as
\begin{eqnarray}
\tau _{0}(x,t) &=&\prod_{j=1}^{N} e^{(\xi _{j}+\bar{\xi}_{j})}\left\vert \frac{1}{p_{i}+%
\bar{p}_{j}}(1+e^{-(p_{i}+\bar{p}_{j})x+\mathrm{i}(p_{i}^{2}-\bar{p%
}_{j}^{2})t+\mathrm{i}(\theta _{i}+\bar{\theta }_{j})})\right\vert\,, \nonumber \\
&=&\prod_{j=1}^{N} e^{\mathrm{i}(\theta _{j}+\bar{\theta }_{j})}\left\vert \frac{1%
}{p_{i}+\bar{p}_{j}}(1+e^{(p_{i}+\bar{p}_{j})x-\mathrm{i}(p_{i}^{2}-%
\bar{p}_{j}^{2})t-\mathrm{i}(\theta _{i}+\bar{\theta }%
_{j})})\right\vert\,,
\end{eqnarray}
which implies
\begin{equation}
\tau _{0}^{\ast }(-x,t)= \prod_{j=1}^{N} e^{-(\xi_{j}+\bar{\xi}_{j})} \left\vert \frac{1}{p_{i}+\bar{p}_{j}}%
(1+e^{(p_{i}+\bar{p}_{j})x-\mathrm{i}(p_{i}^{2}-\bar{p}_{j}^{2})t-%
\mathrm{i}(\theta _{i}+\bar{\theta }_{j})})\right\vert\,.
\end{equation}
Therefore $\tau _{0}(x,t)=C \tau^{\ast} _{0}(-x,t)$.
On the other hand,
\begin{eqnarray}
\tau _{1}(x,t) &=&\left\vert
\begin{array}{cc}
\frac{1}{p_{i}+\bar{p}_{j}}(e^{(p_{i}+\bar{p}_{j})x-\mathrm{i}%
(p_{i}^{2}-\bar{p}_{j}^{2})t}+e^{\mathrm{i}(\theta _{i}+\bar{%
\theta }_{j})}) & e^{p_{i}x-{ \mathrm{i}}p_{i}^{2}t} \\
-e^{\mathrm{i}\bar{\theta }_{j}} & 0%
\end{array}%
\right\vert\,, \nonumber \\
&=&\prod_{j=1}^{N} e^{\mathrm{i}(\theta _{j}+\bar{\theta }_{j})}\left\vert
\begin{array}{cc}
\frac{1}{p_{i}+\bar{p}_{j}}(1+e^{\xi _{i}+\bar{\xi}_{j}-\mathrm{i}(\theta _{i}+\bar{\theta}_j)}) & e^{\xi _{i} -\mathrm{i}\theta _{i}} \\
-1 & 0%
\end{array}%
\right\vert\,, \nonumber \\
&=&\prod_{j=1}^{N} e^{(\xi _{j}+\bar{\xi}_{j})}\left\vert
\begin{array}{cc}
\frac{1}{p_{i}+\bar{p}_{j}}(1+e^{-\xi _{i}-\bar{\xi}_{j}+\mathrm{i}(\theta _{i}+\bar{\theta }_{j})}) & 1 \\
-e^{-\bar{\xi}_{j}+\mathrm{i}%
\bar{\theta }_{j}} & 0%
\end{array}%
\right\vert\,,
\end{eqnarray}

\begin{eqnarray}
\tau _{-1}(x,t) &=&\left\vert
\begin{array}{cc}
\frac{1}{p_{i}+\bar{p}_{j}}(e^{(p_{i}+\bar{p}_{j})x-\mathrm{i}%
(p_{i}^{2}-\bar{p}_{j}^{2})t}+e^{\mathrm{i}(\theta _{i}+\bar{%
\theta }_{j})}) & e^{{\mathrm{i}}\theta _{i}} \\
-e^{\bar{p}_{j}x+{\mathrm{i}}\bar{p}_{j}^{2}t} & 0%
\end{array}%
\right\vert\,, \nonumber \\
&=&\prod_{j=1}^{N} e^{(\xi _{j}+\bar{\xi}_{j})}\left\vert
\begin{array}{cc}
\frac{1}{p_{i}+\bar{p}_{j}}(1+e^{-\xi _{i}-\bar{\xi}_{j}+\mathrm{i}(\theta _{i}+\bar{\theta }_{j})}) & e^{-\xi _{i} +\mathrm{i}\theta _{i}} \\
-1 & 0%
\end{array}%
\right\vert\,.
\end{eqnarray}
Therefore
\begin{eqnarray}
\tau^{\ast}_{1}(-x,t)
&=&\prod_{j=1}^{N} e^{(\xi _{j}+\bar{\xi}_{j})^{\ast}(-x,t) }\left\vert
\begin{array}{cc}
\frac{1}{p_{i}+\bar{p}_{j}}(1+e^{\xi _{i}+\bar{\xi}_{j}-\mathrm{i}(\theta _{i}+\bar{\theta }_{j})}) & 1 \\
-e^{\bar{\xi}_{j}-\mathrm{i}%
\bar{\theta }_{j}} & 0%
\end{array}%
\right\vert\,, \nonumber \\
&=& \prod_{j=1}^{N} e^{-\mathrm{i}(\theta _{j}+\bar{\theta }_{j})}\left\vert
\begin{array}{cc}
\frac{1}{p_{i}+\bar{p}_{j}}(1+e^{-\xi _{i}-\bar{\xi}_{j}+\mathrm{i}(\theta _{i}+\bar{\theta}_j)}) & e^{-\xi_i+\mathrm{i}\theta _{i}} \\
-1 & 0%
\end{array}%
\right\vert\,.
\end{eqnarray}
Obviously $\tau_{-1}(x,t)=C \tau^{\ast} _{1}(-x,t)$.
Next we prove the case where there are pairs of wave numbers $p_{k},p_{k^{\prime}},%
\bar{p}_{k},\bar{p}_{k^{\prime }}$ such that $p_{k}=p_{k^{\prime }}^{\ast},%
\bar{p}_{k}=\bar{p}_{k^{\prime }}^{\ast }$. Moreover, ${\eta }%
_{k^{\prime }0}=-{\eta }_{k0}^{\ast }$, $\bar{\eta}_{k^{\prime }0}=-\bar{\eta}_{k0}^{\ast }$.
Note that $p_{k},p_{k^{\prime}}$ or $\bar{p}_{k},\bar{p}_{k^{\prime }}$ being real; ${\eta}_{k^{\prime }0}$, ${\eta }_{k0}$ or $\bar{\eta}_{k^{\prime }0}$, $\bar{\eta}_{k0}$ being purely imaginary is simply a special case. Under this case, we
also have $\xi _{k}^{\ast }(-x,t)=-\xi _{k'}(x,t),\bar{\xi}_{k}^{\ast
}(-x,t)=-\bar{\xi}_{k'}(x,t)$, thus
\begin{eqnarray}
\fl \tau _{0}(x,t) &=&\left\vert
\begin{array}{cc}
\frac{1}{p_{k}+\bar{p}_{k^{\prime }}}(e^{\xi _{k}+\bar{\xi}_{k^{\prime
}}}+e^{{\eta }_{k0}+\bar{\eta}_{k^{\prime }0}}) & \cdots  \\
\cdots  & \frac{1}{p_{k^{\prime }}+\bar{p}_{k}}(e^{\xi _{k^{\prime }}+\bar{%
\xi}_{k}}+e^{{\eta }_{k^{\prime }0}+\bar{\eta}_{k0}})%
\end{array}%
\right\vert  \nonumber \\
\fl &=&C_{2}\left\vert
\begin{array}{cc}
\frac{1}{p_{k}+\bar{p}_{k^{\prime }}}(1+e^{\xi _{k}+\bar{\xi}_{k^{\prime }}-{%
\eta }_{k0}-\bar{\eta}_{k^{\prime }0}}) & \cdots  \\
\cdots  & \frac{1}{p_{k^{\prime }}+\bar{p}_{k}}(1+e^{\xi _{k^{\prime }}+\bar{%
\xi}_{k}-{\eta }_{k^{\prime }0}+\bar{\eta}_{k0}}%
\end{array}%
\right\vert \nonumber \\
\fl &=&C_{1}\left\vert
\begin{array}{cc}
\frac{1}{p_{k}+\bar{p}_{k^{\prime }}}(1+e^{-(\xi _{k}+\bar{\xi}_{k^{\prime
}}-{\eta }_{k0}-\bar{\eta}_{k^{\prime }0})}) & \cdots  \\
\cdots  & \frac{1}{p_{k^{\prime }}+\bar{p}_{k}}(1+e^{-(\xi _{k^{\prime }}+%
\bar{\xi}_{k}-{\eta }_{k^{\prime }0}-\bar{\eta}_{k0})})%
\end{array}%
\right\vert \nonumber  \\
\fl &=&C_{1}\left\vert
\begin{array}{cc}
\frac{1}{p_{k^{\prime }}+\bar{p}_{k}}(1+e^{-(\xi _{k^{\prime }}+\bar{\xi}%
_{k}-{\eta }_{k^{\prime }0}-\bar{\eta}_{k0})}) & \cdots  \\
\cdots  & \frac{1}{p_{k}+\bar{p}_{k^{\prime }}}(1+e^{-(\xi _{k}+\bar{\xi}%
_{k^{\prime }}-{\eta }_{k0}-\bar{\eta}_{k^{\prime }0})})%
\end{array}%
\right\vert\,.
\end{eqnarray}%
Therefore
\begin{equation}
\fl \tau _{0}^{\ast }(-x,t)=C_{1}^{-1}\left\vert
\begin{array}{cc}
\frac{1}{p_{k}+\bar{p}_{k^{\prime }}}(1+e^{\xi _{k}+\bar{\xi}_{k^{\prime }}-{%
\eta }_{k0}-\bar{\eta}_{k0})}) & \cdots  \\
\cdots  & \frac{1}{p_{k^{\prime }}+\bar{p}_{k}}(1+e^{-(\xi _{k}+\bar{\xi}%
_{k^{\prime }}-{\eta }_{k0}-\bar{\eta}_{k^{\prime }0})})%
\end{array}%
\right\vert \,.
\end{equation}
Obviously $\tau_{0}(x,t)=C \tau^{\ast} _{0}(-x,t)$. On the other hand,
\begin{eqnarray*}
\fl \tau _{1}(x,t) &=&\left\vert
\begin{array}{ccc}
\frac{1}{p_{k}+\bar{p}_{k^{\prime }}}(e^{\xi _{k}+\bar{\xi}_{k^{\prime
}}}+e^{{\eta }_{k0}+\bar{\eta}_{k^{\prime }0}}) & \cdots  & e^{\xi _{k}} \\
\cdots  & \frac{1}{p_{k^{\prime }}+\bar{p}_{k}}(e^{\xi _{k^{\prime }}+\bar{%
\xi}_{k}}+e^{{\eta }_{k^{\prime }0}+\bar{\eta}_{k0}}) & e^{\xi _{k^{\prime
}}} \\
-e^{\bar{\eta}_{k^{\prime }0}} & -e^{\bar{\eta}_{k0}} & 0%
\end{array}%
\right\vert \nonumber \\
\fl &=&C_{2}\left\vert
\begin{array}{ccc}
\frac{1}{p_{k}+\bar{p}_{k^{\prime }}}(1+e^{\xi _{k}+\bar{\xi}_{k^{\prime }}-{%
\eta }_{k0}-\bar{\eta}_{k^{\prime }0}}) & \cdots  & e^{\xi _{k}-{\eta }_{k0}}
\\
\cdots  & \frac{1}{p_{k^{\prime }}+\bar{p}_{k}}(1+e^{\xi _{k^{\prime }}+\bar{%
\xi}_{k}-{\eta }_{k^{\prime }0}-\bar{\eta}_{k0}}) & e^{\xi _{k^{\prime}}-{%
\eta }_{k^{\prime }0}} \\
-1 & -1 & 0%
\end{array}%
\right\vert \nonumber \\
\fl &=&C_{1}\left\vert
\begin{array}{ccc}
\frac{1}{p_{k}+\bar{p}_{k^{\prime }}}(1+e^{-(\xi _{k}+\bar{\xi}_{k^{\prime
}}-{\eta }_{k0}-\bar{\eta}_{k^{\prime }0})}) & \cdots  & 1 \\
\cdots  & \frac{1}{p_{k^{\prime }}+\bar{p}_{k}}(1+e^{-(\xi _{k^{\prime }}+%
\bar{\xi}_{k}-{\eta }_{k^{\prime }0}-\bar{\eta}_{k0})}) & 1 \\
-e^{\bar{\eta}_{k^{\prime }0}} & -e^{\bar{\eta}_{k0}} & 0%
\end{array}%
\right\vert \nonumber \\
\fl &=&C_{1}\left\vert
\begin{array}{ccc}
\frac{1}{p_{k^{\prime }}+\bar{p}_{k}}(1+e^{-(\xi _{k^{\prime }}+\bar{\xi}%
_{k}-{\eta }_{k^{\prime }0}-\bar{\eta}_{k0})}) & \cdots  & 1 \\
\cdots  & \frac{1}{p_{k}+\bar{p}_{k^{\prime }}}(1+e^{-(\xi _{k}+\bar{\xi}%
_{k^{\prime }}-{\eta }_{k0}-\bar{\eta}_{k^{\prime }0})}) & 1 \\
-e^{\bar{\eta}_{k0}} & -e^{\bar{\eta}_{k^{\prime }0}} & 0%
\end{array}%
\right\vert\,,
\end{eqnarray*}

\begin{eqnarray*}
\fl \tau _{-1}(x,t) &=&\left\vert
\begin{array}{ccc}
\frac{1}{p_{k}+\bar{p}_{k^{\prime }}}(e^{\xi _{k}+\bar{\xi}_{k^{\prime
}}}+e^{{\eta }_{k0}+\bar{\eta}_{k^{\prime }0}}) & \cdots  & e^{{\eta }_{k0}}
\\
\cdots  & \frac{1}{p_{k^{\prime }}+\bar{p}_{k}}(e^{\xi _{k^{\prime }}+\bar{%
\xi}_{k}}+e^{{\eta }_{k^{\prime }0}+\bar{\eta}_{k0}}) & e^{{\eta }%
_{k^{\prime }0}} \\
-e^{\bar{\xi}_{k^{\prime }}} & -e^{\bar{\xi}_{k}} & 0%
\end{array}%
\right\vert  \\
\fl &=&C_{1}\left\vert
\begin{array}{ccc}
\frac{1}{p_{k^{\prime }}+\bar{p}_{k}}(1+e^{-(\xi _{k^{\prime }}+\bar{\xi}%
_{k}-{\eta }_{k^{\prime }0}-\bar{\eta}_{k0})}) & \cdots  & e^{-\xi
_{k^{\prime }}+{\eta }_{k^{\prime }0}} \\
\cdots  & \frac{1}{p_{k}+\bar{p}_{k^{\prime}}}(1+e^{-(\xi _{k}+\bar{\xi}%
_{k^{\prime }}-{\eta }_{k0}-\bar{\eta}_{k^{\prime }0})}) & e^{-\xi _{k}+{%
\eta }_{k0}} \\
-1 & -1 & 0%
\end{array}%
\right\vert\,.
\end{eqnarray*}
Thus
\begin{equation*}
\fl \tau _{-1}^{\ast }(-x,t)=C_{1}^{-1}\left\vert
\begin{array}{ccc}
\frac{1}{p_{k}+\bar{p}_{k^{\prime }}}(1+e^{\xi _{k}+\bar{\xi}_{k^{\prime }}-{%
\eta }_{k0}-\bar{\eta}_{k^{\prime }0}}) & \cdots  & e^{\xi _{k}-{\eta }_{k0}}
\\
\cdots  & \frac{1}{p_{k^{\prime }}+\bar{p}_{k}}(1+e^{\xi _{k^{\prime }}+\bar{%
\xi}_{k}-{\eta }_{k^{\prime }0}-\bar{\eta}_{k0}}) & e^{\xi _{k^{\prime
}}-{\eta }_{k^{\prime }0}} \\
-1 & -1 & 0%
\end{array}%
\right\vert\,.
\end{equation*}
Consequently, $\tau _{-1}(-x,t)=C \tau _{1}^{\ast }(x,t)$.
\end{proof}
Following above Lemma, if we define $\tau_0(x,t)=\sqrt{C}f(x,t)$, $\tau_{1}(x,t)=\sqrt{C}g(x,t)$, $\tau_{-1}(x,t)=\sqrt{C}\bar{g}(x,t)$, then we have $f(x,t)=f^{\ast}(-x,t)$, $\bar{g}(x,t)=g^{\ast}(-x,t)$, which lead to the following bilinear equations
\begin{equation}
\label{bilinear_bright2}
\left\{
\begin{array}{l}
\displaystyle(\mathrm{i} D_{t}-D_{x}^{2})g(x,t) \cdot f(x,t)=0\,, \\[5pt]
\displaystyle D^2_{x}f(x,t) \cdot f(x,t) -2g(x,t) g^{\ast} (-x,t)=0\,.%
\end{array}%
\right.
\end{equation}
Moreover, if we define
\begin{equation}
\label{NLS_brightGram}
q(x,t)=\frac{g(x,t)}{f(x,t)}, \quad r(x,t)= \frac{-\bar{g}(x,t)}{f(x,t)}\,,
\end{equation}
we then have
\begin{equation}
r(x,t)=-q^{\ast }(-x,t).
\end{equation}
In summary, we have the general $N$-bright soliton solution (\ref{NLS_brightGram}) to the nonlocal NLS equation
(\ref{NLS_PT}) with $\sigma=-1$.
\begin{equation}
f(x,t)= \frac{1}{\sqrt{C}}\left\vert \frac{1}{p_{i}+\bar{p}_{j}}(e^{(p_{i}+\bar{p}_{j})x-\mathrm{i}%
(p_{i}^{2}-\bar{p}_{j}^{2})t}+e^{(\eta _{i}+\bar{%
\eta }_{j})})\right\vert_{N \times N},,
\end{equation}
\begin{equation}
g(x,t) =\frac{1}{\sqrt{C}} \left\vert
\begin{array}{cc}
\frac{1}{p_{i}+\bar{p}_{j}}(e^{(p_{i}+\bar{p}_{j})x-\mathrm{i}%
(p_{i}^{2}-\bar{p}_{j}^{2})t}+e^{(\eta _{i}+\bar{%
\eta }_{j})}) & e^{p_{i}x-{ \mathrm{i}}p_{i}^{2}t} \\
-e^{\bar{\eta }_{j}} & 0%
\end{array}%
\right\vert_{(N+1) \times (N+1)}\,.
\end{equation}
In what follows, we list one- and two-soliton solutions: \\
{\bf {One-soliton solution:}}
\begin{equation}
f(x,t)=\frac{1}{%
p_{1}+\bar{p}_{1}} e^{-\frac 12 ({\xi _{1}+\bar{\xi}_{1}}-\mathrm{i} \theta _{1}-\mathrm{i}\bar{\theta }_{1})} \left( 1+e^{(p_{1}+\bar{p}_{1})x-\mathrm{i}(p_{1}^{2}-%
\bar{p}_{1}^{2})t-\mathrm{i}(\theta _{1}+\bar{\theta }%
_{1})}\right)
\end{equation}

\begin{equation}
g(x,t)=e^{-\frac 12 ({\xi _{1}+\bar{\xi}_{1}}+ \mathrm{i} \theta _{1}+ \mathrm{i}\bar{\theta }_{1})} e^{p_{1}x-\mathrm{i} p_{1}^{2}t+\mathrm{i}\bar{\theta }_{1}}\,.
\end{equation}
So
\begin{equation}
q(x,t)=\frac{(p_{1}+\bar{p}_{1})e^{p_{1}x- \mathrm{i} p_{1}^{2}t-\mathrm{i}{\theta}_{1}}}
{1+e^{(p_{1}+\bar{p}_{1})x-\mathrm{i}(p_{1}^{2}-%
\bar{p}_{1}^{2})t-\mathrm{i}(\theta _{1}+\bar{\theta }%
_{1})}}\,.
\end{equation}
 {if we let $p_1=-2\bar{\eta}$, $\bar{p}_1=-2 \eta$, $\theta_1=-\bar{\theta}+\pi$, $\bar{\theta}_1=-\theta$, then above solution exactly recovers the one-soliton solution found in \cite{AblowitzMusslimani3}, also mentioned in the introduction.}\\
{\bf {Two-soliton solution:}}
\begin{eqnarray}
\fl f &=&\left\vert
\begin{array}{cc}
\frac{1}{p_{1}+\bar{p}_{1}}(e^{\xi _{1}+\bar{\xi}_{1}}+e^{{\eta }_{10}+\bar{%
\eta}_{10}}) & \frac{1}{p_{1}+\bar{p}_{2}}(e^{\xi _{1}+\bar{\xi}_{2}}+e^{{%
\eta }_{10}+\bar{\eta}_{20}}) \\
\frac{1}{p_{2}+\bar{p}_{1}}(e^{\xi _{2}+\bar{\xi}_{1}}+e^{{\eta }_{20}+\bar{%
\eta}_{10}}) & \frac{1}{p_{2}+\bar{p}_{2}}(e^{\xi _{2}+\bar{\xi}_{2}}+e^{{%
\eta }_{20}+\bar{\eta}_{20}})%
\end{array}%
\right\vert \nonumber  \\
\fl &=&D \left( 1+e^{\xi _{1}+\bar{\xi}_{1}+\xi _{2}+\bar{\xi}_{2}-{\eta }_{10}-%
\bar{\eta}_{10}-{\eta }_{20}-\bar{\eta}_{20}}+\frac{(p_{1}+\bar{p}%
_{2})(p_{2}+\bar{p}_{1})}{(p_{1}-p_{2})(\bar{p}_{1}-\bar{p}_{2})}\left(
e^{\xi _{1}+\bar{\xi}_{1}-{\eta }_{10}-\bar{\eta}_{10}}+e^{\xi _{2}+\bar{\xi}%
_{2}-{\eta }_{20}-\bar{\eta}_{20}}\right) \right. \nonumber \\
\fl && \left. +\frac{(p_{1}+\bar{p}_{1})(p_{2}+\bar{p}_{2})}{(p_{1}-p_{2})(\bar{p}_{1}-%
\bar{p}_{2})}\left( e^{\xi _{1}+\bar{\xi}_{2}-{\eta }_{10}-\bar{\eta}%
_{20}}+e^{\xi _{2}+\bar{\xi}_{1}-{\eta }_{20}-\bar{\eta}_{10}}\right) \right)\,,
\end{eqnarray}

\begin{eqnarray}
\fl g &=&\left\vert
\begin{array}{ccc}
\frac{1}{p_{1}+\bar{p}_{1}}(e^{\xi _{1}+\bar{\xi}_{1}}+e^{{\eta }_{10}+\bar{%
\eta}_{10}}) & \frac{1}{p_{1}+\bar{p}_{2}}(e^{\xi _{1}+\bar{\xi}_{2}}+e^{{%
\eta }_{10}+\bar{\eta}_{20}}) & e^{\xi _{1}} \\
\frac{1}{p_{2}+\bar{p}_{1}}(e^{\xi _{2}+\bar{\xi}_{1}}+e^{{\eta }_{20}+\bar{%
\eta}_{10}}) & \frac{1}{p_{2}+\bar{p}_{2}}(e^{\xi _{2}+\bar{\xi}_{2}}+e^{{%
\eta }_{20}+\bar{\eta}_{20}}) & e^{\xi _{2}} \\
-e^{\bar{\eta}_{10}} & -e^{\bar{\eta}_{20}} & 0%
\end{array}%
\right\vert  \nonumber \\
\fl &=&D\left( \frac{(p_{1}+\bar{p}_{1})(p_{2}+\bar{p}_{1})}{\bar{p}_{1}-\bar{p}%
_{2}}e^{\xi _{1}+\xi _{2}+\bar{\xi}_{2}-{\eta }_{10}-{\eta }_{20}-\bar{\eta}%
_{20}}+\frac{(p_{1}+\bar{p}_{1})(p_{1}+\bar{p}_{2})}{p_{1}-p_{2}}e^{\xi _{1}-%
{\eta }_{10}}\right) \nonumber \\
\fl &&-\left( \frac{(p_{1}+\bar{p}_{2})(p_{2}+\bar{p}_{2})}{\bar{p}_{1}-\bar{p}%
_{2}}e^{\xi _{1}+\xi _{2}+\bar{\xi}_{2}-{\eta }_{10}-{\eta }_{20}-\bar{\eta}%
_{20}}+\frac{(p_{2}+\bar{p}_{1})(p_{2}+\bar{p}_{2})}{p_{1}-p_{2}}e^{\xi _{2}-%
{\eta }_{20}}\right)\,,
\end{eqnarray}
where
\begin{equation*}
D=\sqrt{\frac{C_2}{C_1}}\frac{(p_{1}-p_{2})(\bar{p}_{1}-\bar{p}_{2})}{(p_{1}+\bar{p}_{1})(p_{1}+%
\bar{p}_{2})(p_{2}+\bar{p}_{1})(p_{2}+\bar{p}_{2})}\,.
\end{equation*}
Two soliton solution could have the following cases:
\begin{enumerate}
  \item $p_{1},p_{2},\bar{p}_{1},\bar{p}_{2}$ are all real and ${\eta }_{10},{%
\eta }_{20},\bar{\eta}_{10},\bar{\eta}_{20}$ are all purely imaginary.
  \item $p_{1},p_{2},\bar{p}_{1},\bar{p}_{2}$, ${\eta }_{10},{\eta }_{20},\bar{%
\eta}_{10},\bar{\eta}_{20}$ are all complex, $p_{1}=p_{2}^{\ast },\bar{p}%
_{1}=\bar{p}_{2}^{\ast },{\eta }_{10}=-{\eta }_{20}^{\ast },\bar{\eta}_{10}=%
\bar{\eta}_{20}^{\ast }$.
  \item $p_{1},p_{2}$, ${\eta }_{10},{\eta }_{20}$ are complex number, $%
p_{1}=p_{2}^{\ast },{\eta }_{10}=-{\eta }_{20}^{\ast };$ $\bar{p}_{1},\bar{p}%
_{2}$ are real and $\bar{\eta}_{10},\bar{\eta}_{20}$ are all purely imaginary.
  \item $\bar{p}_{1},\bar{p}_{2}$, $\bar{\eta}_{10},\bar{\eta}_{20}$ are complex
number, $\bar{p}_{1}=\bar{p}_{2}^{\ast },\bar{\eta}_{10}=\bar{\eta}%
_{20}^{\ast };$ $p_{1},p_{2}$ are real and ${\eta }_{10},{\eta }_{20}$ are
all purely imaginary.
\end{enumerate}
It is noted that all above four cases regarding the bright soliton solution to the nonlocal NLS equation have been recognized and their dynamics have been detailed discussed in \cite{JYangNonlocalNLS}.
\subsection{General bright soliton solution expressed by double Wronskian determinant}
Alternatively, we can also present the general bright soliton solution to the nonlocal NLS equation (\ref{NLS_PT}) in terms of the double Wronskian determinant. To this end, we start with
the tau-functions for two-component KP hierarchy expressed in double Wronskian determinant
\begin{equation}
\tau _{N,M}=\left\vert
\begin{array}{cccccc}
\phi _{1}^{(0)} & \cdots & \phi _{1}^{(N-1)} & \psi _{1}^{(0)} & \cdots &
\psi _{1}^{(M-1)} \\
\phi _{2}^{(0)} & \cdots & \phi _{2}^{(N-1)} & \psi _{2}^{(0)} & \cdots &
\psi _{2}^{(M-1)} \\
\vdots & \vdots & \vdots & \vdots & \vdots & \vdots \\
\phi _{N+M}^{(0)} & \cdots & \phi _{N+M}^{(N-1)} & \psi _{N+M}^{(0)} & \cdots
& \psi _{N+M}^{(M-1)}%
\end{array}%
\right\vert _{(N+M)\times (N+M)}\,,
\end{equation}
here $\phi _{i}^{(n)}$ and $\psi _{i}^{(n)}$ take the form
\[
\phi _{i}^{(n)}=p_{i}^{n}e^{\xi _{i}}, \quad \psi _{i}^{(n)}=q_{i}^{n}e^{\eta
_{i}}\,,
\]%
with
\begin{equation*}
\xi _{i}=p_{i}x_{1}+p_{i}^{2}x_{2}+\xi _{0i}+\cdots\,, \quad \eta _{i}=q_{i}y_{1}+\eta _{0i}+\cdots\,.
\end{equation*}
The above tau functions satisfy the following bilinear equations
\begin{equation}
\left\{
\begin{array}{l}
\displaystyle(D_{x_{2}}-D_{x_{1}}^{2})\tau _{N+1,N-1}\cdot \tau _{N,N}=0, \\%
[5pt]
\displaystyle(D_{x_{2}}-D_{x_{1}}^{2})\tau _{N,N}\cdot \tau
_{N-1,N+1}=0,\quad \\
\displaystyle D_{x_{1}}D_{y_{1}}\tau _{N,N}\cdot \tau _{N,N}-2\tau
_{N+1,N-1}\tau _{N-1,N+1}=0\,.%
\end{array}%
\right.
\label{bilinear_doubleWr}
\end{equation}
The proof of above equations can be done by using the determinant
technique \cite{HirotaBook},which is omitted here. In what follows, we will
perform reduction procedure from bilinear equations (\ref{bilinear_doubleWr}) to the bilinear equations (\ref{NLSPT_bilinear}). For the sake of convenience,
we take the following abbraevation
\[
\tau _{N,M}\,=\left\vert 0,\cdots ,N-1,0^{^{\prime }},\cdots
,(M-1)^{^{\prime }}\right\vert\,,
\]
First, we consider the dimension reduction. Imposing the condition $N=M$ and  $%
q_{i}=p_{i}$ $(i=1,\cdots ,2N)$, we then have
\begin{equation}
(\partial _{x_{1}}+\partial _{y_{1}})\tau _{N,M}=\left(
\prod_{j=1}^{2N}p_{j}\right) \tau _{N,M}\,.
\end{equation}%
Under this reduction, $y_{1}$ becomes a dummy variable, which can be taken
as zero. Applying variable transformations
\begin{equation}
x_{1}=x,\quad x_{2}=-\mathrm{i}t\,,  \label{var-trf}
\end{equation}%
which implies
\begin{equation}
\partial _{x_{1}}=\partial _{x},\quad \partial _{x_{2}}=\mathrm{i}\partial
_{t}\,.  \label{var-trf2}
\end{equation}
Then the bilinear equations (\ref{bilinear_doubleWr}) can be rewritten as
\begin{equation}
\left\{
\begin{array}{l}
\displaystyle(\mathrm{i}D_{t}-D_{x}^{2})\tau _{N+1,N-1}\cdot \tau _{N,N}=0,
\\[5pt]
\displaystyle(\mathrm{i}D_{t}-D_{x}^{2})\tau _{N,N}\cdot \tau
_{N-1,N+1}=0,\quad  \\
\displaystyle D_{x}^{2}\tau _{N,N}\cdot \tau _{N,N}+2\tau _{N+1,N-1}\tau
_{N-1,N+1}=0\,,%
\end{array}%
\right.
\end{equation}

\begin{lemma}
Assume $x_{1}=x$ is real and $x_{2}=-\mathrm{i}t$ is purely imaginary variable. Suppose there are $K$ pairs of wave
numbers ($p_{k}$, ${\eta }_{k0}$) and ($\bar{p}_{k}$, $\bar{\eta}_{k0}$) for $k=1, \cdots, K$ such that
 $p_{k}=p^*_{k'}$, ${\eta }_{k0}={\eta}^*_{k'0}$ or $\bar{p}_{k}=\bar{p}^*_{k'}$, $\bar{\eta}_{k0}=\bar{\eta}^*_{k'0}$; and the rest of wave numbers $p_{j}$ and $\bar{p}_{j}$ are real, ${\eta }_{j0}$, $\bar{\eta}_{j0}$ are purely imaginary. Then
\begin{equation}
\label{tau_relationDW1}
\tau _{N,N}(x,t) = C' (-1)^{N+K} \tau^{\ast} _{N,N}(-x,t)\,,
\end{equation}
\begin{equation}
\label{tau_relationDW2}
\tau _{N-1,N+1}(x,t) = C' (-1)^{N+K+1} \tau^{\ast} _{N+1,N-1}(-x,t)\,,
\end{equation}
where $C'=\prod^{N}_{i=1} e^{\xi_{i}(x,t)-\eta _{i}}$.
\end{lemma}
\begin{proof}
We only give the proof when $p_{j}$ and $\bar{p}_{j}$ are real and ${\eta }_{j0}$, $\bar{\eta}_{j0}$ are purely imaginary. Under this case,
 $\xi _{i}(x,t)=p_{i}x-\mathrm{i}p_{i}^{2}t=-\xi _{i}^{\ast}(-x,t),$ then it can be shown
\begin{eqnarray}
\fl \tau _{N-1,N+1}(x,t) &=& \left\vert 0,\cdots ,N-2,0^{^{\prime }}\cdots
,,N^{^{\prime }}\right\vert \,, \nonumber \\
&=&\left\vert e^{\xi _{i}(x,t)},\cdots ,p_{i}^{N-2}e^{^{\xi
_{i}(x,t)}},e^{-\mathrm{i}\theta _{i}},\cdots ,p_{i}^{N}e^{-\mathrm{i}\theta
_{i}}\right\vert \,, \nonumber \\
\fl &=& \prod^{N}_{i=1} e^{\xi _{i}(x,t)-\mathrm{i}\theta _{i}}\left\vert e^{%
\mathrm{i}\theta _{i}},\cdots ,p_{i}^{N-2}e^{\mathrm{i}\theta _{i}},e^{-\xi
_{i}(x,t)},\cdots ,p_{i}^{N}e^{-\xi _{i}(x,t)}\right\vert \,, \nonumber \\
\fl &=& (-1)^N \prod^{N}_{i=1} e^{\xi _{i}(x,t)-\mathrm{i}\theta _{i}}
\left\vert e^{-\xi _{i}(x,t)},\cdots ,p_{i}^{N}e^{-\xi
_{i}(x,t)},e^{\mathrm{i}\theta _{i}},\cdots ,p_{i}^{N-2}e^{\mathrm{i}\theta
_{i}}\right\vert \,, \nonumber \\
\fl &=&(-1)^N \prod^{N}_{i=1} e^{\xi _{i}(x,t)-\mathrm{i}\theta _{i}}
\left\vert e^{\xi _{i}^{\ast }(-x,t)},\cdots
,p_{i}^{N}e^{\xi _{i}^{\ast }(-x,t)},e^{\mathrm{i}\theta _{i}},\cdots
,p_{i}^{N-2}e^{\mathrm{i}\theta _{i}}\right\vert\,, \nonumber\\
\fl  &=& (-1)^{N+1} C' \tau^{\ast} _{N+1,N-1}(-x,t)\,.
\end{eqnarray}
\begin{eqnarray}
\fl \tau _{N,N}(x,t)&=& \left\vert 0,\cdots ,N-1,0^{^{\prime }},\cdots
,(N-1)^{^{\prime }}\right\vert\,,  \nonumber \\
\fl &=&\left\vert e^{\xi _{i}(x,t)},\cdots ,p_{i}^{N-1}e^{\xi
_{i}(x,t)},e^{-\mathrm{i}\theta _{i}},\cdots ,p_{i}^{N-1}e^{-\mathrm{i}%
\theta _{i}}\right\vert\,,  \nonumber \\
\fl &=&\prod^{N}_{i=1} e^{\xi_{i}(x,t)-\mathrm{i}\theta _{i}}
\left\vert e^{\mathrm{i}\theta _{i}},\cdots
,p_{i}^{N-1}e^{\mathrm{i}\theta _{i}},e^{-\xi _{i}(x,t)},\cdots
,p_{i}^{N-1}e^{-\xi _{i}(x,t)}\right\vert\,,  \nonumber \\
\fl &=&(-1)^{N+1}\prod^{N}_{i=1} e^{\xi_{i}(x,t)-\mathrm{i}\theta _{i}} \left\vert e^{-\xi _{i}(x,t)},\cdots ,p_{i}^{N-1}e^{-\xi
_{i}(x,t)},e^{\mathrm{i}\theta _{i}},\cdots ,p_{i}^{N-1}e^{\mathrm{i}\theta
_{i}}\right\vert\,,  \nonumber \\
\fl &=& (-1)^{N}\prod^{N}_{i=1} e^{\xi_{i}(x,t)-\mathrm{i}\theta _{i}} {\left\vert e^{\xi
_{i}^{\ast }(-x,t)},\cdots ,p_{i}^{N-1}e^{\xi _{i}^{\ast }(-x,t)},e^{\mathrm{%
i}\theta _{i}},\cdots ,p_{i}^{N-1}e^{\mathrm{i}\theta _{i}}\right\vert} \,,  \nonumber \\
\fl &=& (-1)^{N+1} C' \tau^{\ast} _{N,N}(-x,t)\,.
\end{eqnarray}
For the case of complex conjuate wavenumbers pairs, it can be approved the same way as the Lemma in previous subsection. So we omit the proof here.
\end{proof}
Therefore if we define $\tau_{N,N}(x,t)=\sqrt{C'}f(x,t)$, $\tau_{N+1,N-1}(x,t)=\sqrt{C'}g(x,t)$, $\tau _{N-1,N+1}(x,t)=-\sqrt{C'} \bar{g}(x,t)$, then the following bilinear equations follow
\begin{equation}
\label{bilinear_bright3}
\left\{
\begin{array}{l}
\displaystyle(\mathrm{i} D_{t}-D_{x}^{2})g(x,t) \cdot f(x,t)=0\,, \\[5pt]
\displaystyle D^2_{x}f(x,t) \cdot f(x,t) -2g(x,t) g^{\ast} (-x,t)=0\,.%
\end{array}%
\right.
\end{equation}
Moreover, if we define
\begin{equation}
\label{NLS_brightDW}
q(x,t)=\frac{g(x,t)}{f(x,t)}, \quad r(x,t)= \frac{\bar{g}(x,t)}{f(x,t)}\,,
\end{equation}
we then have
\begin{equation}
r(x,t)=-q^{\ast }(-x,t).
\end{equation}
To summarize, we have the general $N$-bright soliton solution
expressed by \begin{equation}
\label{NLS_brightDW2}
q(x,t)=\frac{g'(x,t)}{f'(x,t)},
\end{equation}
 where
\begin{equation}
f'(x,t)=\left\vert
\begin{array}{cccccc}
\tilde\phi _{1}^{(0)} & \cdots & \tilde\phi _{1}^{(N-1)} & \tilde\psi _{1}^{(0)} & \cdots &
\tilde\psi _{1}^{(N-1)} \\
\tilde\phi _{2}^{(0)} & \cdots & \tilde\phi _{2}^{(N-1)} & \tilde\psi _{2}^{(0)} & \cdots &
\tilde\psi _{2}^{(N-1)} \\
\vdots & \vdots & \vdots & \vdots & \vdots & \vdots \\
\tilde\phi _{2N}^{(0)} & \cdots & \tilde\phi _{2N}^{(N-1)} & \tilde\psi _{2N}^{(0)} & \cdots
& \tilde\psi _{2N}^{(N-1)}%
\end{array}%
\right\vert _{2N\times 2N}\,,
\end{equation}
\begin{equation}
g'(x,t)=\left\vert
\begin{array}{cccccc}
\tilde\phi _{1}^{(0)} & \cdots & \tilde\phi _{1}^{(N)} & \tilde\psi _{1}^{(0)} & \cdots &
\tilde\psi _{1}^{(N-2)} \\
\tilde\phi _{2}^{(0)} & \cdots & \tilde\phi _{2}^{(N)} & \tilde\psi _{2}^{(0)} & \cdots &
\tilde\psi _{2}^{(N-2)} \\
\vdots & \vdots & \vdots & \vdots & \vdots & \vdots \\
\tilde\phi _{2N}^{(0)} & \cdots & \tilde\phi _{2N}^{(N)} & \tilde\psi _{2N}^{(0)} & \cdots
& \tilde\psi _{2N}^{(N-2)}%
\end{array}%
\right\vert _{2N\times 2N}\,,
\end{equation}
here {$\tilde\phi _{i}^{(n)}$ and $\tilde\psi _{i}^{(n)}$ take the form
\[
\tilde\phi _{i}^{(n)}=p_{i}^{n}e^{p_{i}x - \mathrm{i} p_{i}^{2} t}, \quad \tilde \psi _{i}^{(n)}=p_{i}^{n}e^{-\mathrm{i} \theta_i}\,.
\]%
}In what follows, we will illustrate
one- and two-soliton solutions and make some comments. By taking $N=1$, we
get the tau functions for one-soliton solution,
\begin{eqnarray}
f'(x,t) &=&\left\vert
\begin{array}{cc}
e^{\xi _{1}} & e^{-\mathrm{i}\theta _{1}} \\
e^{\xi _{2}} & e^{-\mathrm{i}\theta _{2}}%
\end{array}%
\right\vert  =e^{\xi _{1}-\mathrm{i}\theta _{2}}(1-e^{\mathrm{i}\theta _{2}-%
\mathrm{i}\theta _{1}}e^{\xi _{2}-\xi _{1}})
\end{eqnarray}

\begin{eqnarray}
g'(x,t) &=&\left\vert
\begin{array}{cc}
e^{\xi _{1}} & p_{1}e^{\xi _{1}} \\
e^{\xi _{2}} & p_{2}e^{\xi _{2}}%
\end{array}%
\right\vert
=(p_{2}-p_{1})e^{\xi _{2}+\xi _{1}}\,.
\end{eqnarray}
The above tau functions lead to the one-soliton as follows
\begin{eqnarray}
q(x,t) &=&\frac{(p_{2}-p_{1})e^{\xi _{2}+\xi _{1}}}{e^{\xi _{1}-\mathrm{i}\theta
_{2}}(1-e^{\mathrm{i}\theta _{2}-\mathrm{i}\theta _{1}}e^{\xi _{2}-\xi _{1}})%
} =\frac{(p_{2}-p_{1})e^{\xi _{2}+\mathrm{i}\theta _{2}}}{(1-e^{\mathrm{i}%
\theta _{2}-\mathrm{i}\theta _{1}}e^{\xi _{2}-\xi _{1}})}
\end{eqnarray}
 {which can be shown to be exactly the same as the one-soliton solution obtained in previous subsection, thus the same solution found in \cite{AblowitzMusslimani3}.}

\section{Soliton solution to the nonlocal NLS equaotion with nonzero boundary condition}
In this section, we consider the general soliton solution to the nonlocal NLS equation (\ref{NLS_PT}) with  {the same nonzero boundary condition as considered in \cite{AblowitzLuoMusslimani}
\begin{equation}\label{NZBCs}
 q(x,t)\rightarrow \rho e^{\mathrm{i}\theta_{\pm}} \,, \ (\rho>0) \quad   \textrm{as} \ \ x\rightarrow \pm
\infty\,,
\end{equation}
where $\Delta \theta =\theta_+-\theta_-$ is either $0$ or $\pi$.}
Similar to the classical NLS equation, to construct soliton solutions of the nonlocal NLS equation with NZBCs,
we need to start with the tau functions for single-component KP hierarchy expressed in Gram-type determinants
\begin{equation}
\tau _{k}=\left\vert m_{ij}(k)\right\vert _{1\leq i,j\leq N}
\end{equation}
where
\begin{equation}
m_{ij}(k)=c_{ij}+\frac{1}{p_{i}+\bar{p}_{j}}\left( -\frac{p_{i}}{%
\bar{p}_{j}}\right) ^{k}e^{\xi _{i}+\bar{\xi }_{j}}
\end{equation}
\begin{eqnarray*}
\xi _{i} &=&p_{i}^{-1}x_{-1}+p_{i}x_{1}+p_{i}^{2}x_{2}+\xi _{0i}+\cdots \\
\bar{\xi }_{i} &=&\bar{p}_{i}^{-1}x_{-1}+\bar{p}_{i}x_{1}-%
\bar{p}_{i}^{2}x_{2}+\bar{\xi }_{0i}+\cdots
\end{eqnarray*}
Based on the Sato theory, the above tau functions satisfies the following
bilinear equations
\begin{equation}
\label{GramNZ_blinear}
\left\{
\begin{array}{l}
\displaystyle(D_{x_{2}}-D_{x_{1}}^{2})\tau _{k+1}\cdot \tau _{k}=0, \\
\displaystyle(\frac{1}{2}D_{x_{1}}D_{x_{-1}}-1)\tau _{k}\cdot \tau
_{k}=-\tau _{k+1}\tau _{k-1}.%
\end{array}%
\right.
\end{equation}
As stated in the subsequent sections, above two set of bilinear equations will be the key in constructing soliton solutions to the nonlocal NLS equation.  {We also assume $c_{ij}=c_{i}\delta _{ij}$ hereafter.}
\subsection{General soliton solution to the nonlocal NLS equation of $\sigma=1$, {$\Delta \theta=0$} with nonzero boundary condition}
The nonlocal NLS equation (\ref{NLS_PT}) of $\sigma=1$ is converted into a set of bilinear equations
\begin{equation}
\label{ndNLS_bilinear}
\left\{
\begin{array}{l}
\displaystyle(\mathrm{i}D_{t}-D_{x}^{2})g\cdot f=0, \\
\displaystyle(D_{x}^{2}-2 \rho ^{2})f\cdot f=-2 \rho
^{2}g(x,t)g^{\ast }(-x,t)\,,%
\end{array}%
\right.
\end{equation}
via a variable transformation
\begin{equation}
\label{AKNS_trf}
q(x,t)=\rho \frac{g(x,t)}{f(x,t)}e^{\mathrm{i}2 \rho ^{2}t}\,
\end{equation}
under the condition $f(x,t)=f^{\ast} (-x,t)$.
In what follows, we will show how to reduce the bilinear equations (\ref{GramNZ_blinear}) to the bilinear equations
(\ref{ndNLS_bilinear}) by the KP hierarchy reduction method.

Firstly, we perform the dimension reduction. Note that, by row operations, $%
\tau _{k}$ can be rewritten as
\begin{equation}
\tau _{k}=\prod_{j=1}^{N}e^{\xi _{j}+\bar{\xi}_{j}}\left\vert m_{ij}^{\prime
}\right\vert =\prod_{j=1}^{N}e^{\xi _{j}+\bar{\xi}_{j}}\left\vert
c_{i}\delta _{ij}e^{-(\xi _{i}+\bar{\xi }_{j})}+\frac{1}{p_{i}+%
\bar{p}_{j}}\left( -\frac{p_{i}}{\bar{p}_{j}}\right)
^{k}\right\vert \,.
\end{equation}
Since
\begin{equation}
(\partial _{x_{1}}- \rho ^{2}\partial _{x_{-1}})m_{ij}^{{\prime
}}=(p_{i}+\bar{p}_{j})\left( 1-\frac{\rho ^{2}}{p_{i}\bar{p}%
_{j}}\right) c_{i}\delta _{ij}e^{-(\xi _{i}+\bar{\xi }_{j})},
\end{equation}
thus, if we impose the constraints
\begin{equation}
p_{j}\bar{p}_{j}=\rho ^{2},\quad j=1,\cdots ,N\,,
\end{equation}
then
\begin{equation}
\partial _{x_{1}}m_{ij}^{{\prime }}=\rho ^{2}\partial
_{x_{-1}}m_{ij}^{{\prime }},
\end{equation}
which leads to
\begin{equation}
\partial _{x_{1}}\tau _{k}=\rho ^{2}\partial _{x_{-1}}\tau _{k}.
\end{equation}
Therefore, the bilinear equations (\ref{GramNZ_blinear}) become
\begin{equation}
\label{Gram_blinear_dim}
\left\{
\begin{array}{l}
\displaystyle(D_{x_{2}}-D_{x_{1}}^{2})\tau _{k+1}\cdot \tau _{k}=0, \\
\displaystyle(D^2_{x_{1}}-2 \rho^2)\tau _{k}\cdot \tau
_{k}=-2 \rho^2\tau _{k+1}\tau _{k-1}.%
\end{array}%
\right.
\end{equation}
Furthermore, we assume $x_{1}=x$, $x_{2}=-\mathrm{i}t$ and define
\begin{equation}
\tau _{0}(x,t)=Cf(x,t),\
\tau _{1}(x,t)=Cg(x,t),\
\tau _{-1}(x,t)=C\bar{g}(x,t),
\end{equation}
with $C=\prod_{j=1}^{N}e^{\xi _{j}+\bar{\xi}_{j}}$, then the bilinear equations become
\begin{equation}
\label{ndNLS_bilinear2}
\left\{
\begin{array}{l}
\displaystyle(\mathrm{i}D_{t}-D_{x}^{2})g(x,t)\cdot f(x,t)=0, \\
\displaystyle(D_{x}^{2}-2\rho ^{2})f(x,t)\cdot f(x,t)=-2\rho ^{2}g(x,t)\bar{g}(x,t)\,.%
\end{array}%
\right.
\end{equation}
\begin{lemma}
Consider $2N \times 2N$ matrices for $f(x,t)$, $g(x,t)$ and $\bar{g}(x,t)$. If $\bar{p}_{j}=p_{j}^{\ast}$ to be the complex conjugate of $p_{j}$, let $c_{N+j}=-c_{j}^{\ast}$,
$p_{N+j}=-p_{j}$ to be complex, for $j=1,\cdots, N$, and $\xi _{0,N+j}=\xi _{0j}$ to be real, then we have
\begin{equation}
f(x,t)=f^{\ast }(-x,t), \quad \bar{g}(x,t)=g^{\ast }(-x,t).
\end{equation}
\end{lemma}
\begin{proof}
Since
\begin{eqnarray*}
(\xi _{i}+\bar{\xi }_{i})(x,t)=(p_{i}+p_{i}^{\ast
})x-(p_{i}^{2}-(p_{i}^{\ast })^{2})\mathrm{i}t+\xi _{0i}\,,
\end{eqnarray*}
and
\begin{eqnarray*}
(\xi _{N+i}+\bar{\xi}_{N+i})(x,t) &=&(p_{N+i}+p_{N+i}^{\ast
})x-(p_{N+i}^{2}-(p_{N+i}^{\ast })^{2})\mathrm{i}t+\xi _{0,N+i} \\
&=&-(p_{i}+p_{i}^{\ast })x-(p_{i}^{2}-(p_{i}^{\ast })^{2})\mathrm{i}t+\xi
_{0i}\,,
\end{eqnarray*}
so we have
\begin{equation*}
(\xi _{N+i}+\bar{\xi}_{N+i})(x,t)=(\xi _{i}+\bar{\xi }_{i})^{\ast
}(-x,t)\,, \quad (\xi _{i}+\bar{\xi }_{i})(x,t)=(\xi _{N+i}+\bar{\xi}_{N+i})^{\ast
}(-x,t).
\end{equation*}
Note that
\begin{eqnarray}
\fl f(x,t) &=&\left\vert
\begin{array}{cc}
c_{i}\delta _{ij}e^{-(\xi_i+\bar{\xi}_{j})}+\frac{1}{p_{i}+p_{j}^{\ast }} &
\frac{1}{p_{i}+p_{N+j}^{\ast }} \\
\fl \frac{1}{p_{N+i}+p_{j}^{\ast }} & c_{N+i}\delta _{N+i,N+j}e^{-(\xi _{N+i}+%
\bar{\xi}_{N+j})}+\frac{1}{p_{N+i}+p_{N+j}^{\ast }}%
\end{array}%
\right\vert_{1\le i,j \le N} \nonumber \\
\fl &=&\left\vert
\begin{array}{cc}
c_{N+i}\delta _{N+i,N+j}e^{-(\xi _{N+i}+\bar{\xi}_{N+j})}+\frac{1}{%
p_{N+i}+p_{N+j}^{\ast }} & \frac{1}{p_{N+i}+p_{j}^{\ast }} \\
\fl \frac{1}{p_{i}+p_{N+j}^{\ast }} & c_{i}\delta _{ij}e^{-(\xi_i+\bar{\xi}%
_{j})}+\frac{1}{p_{i}+p_{j}^{\ast }}%
\end{array}%
\right\vert \nonumber \\
\fl &=&\left\vert
\begin{array}{cc}
-c_{i}^{\ast }\delta _{ij}e^{-(\xi _{N+i}+\bar{\xi}_{N+j})}-\frac{1}{%
p_{i}+p_{j}^{\ast }} & -\frac{1}{p_{i}+p_{N+j}^{\ast }} \\
-\frac{1}{p_{N+i}+p_{j}^{\ast }} & -c_{N+i}^{\ast }\delta _{i,j}e^{-(\xi_i+%
\bar{\xi}_{j})}-\frac{1}{p_{N+i}+p_{N+j}^{\ast }}%
\end{array}%
\right\vert \nonumber \\
\fl &=&\left\vert
\begin{array}{cc}
c_{i}^{\ast }\delta _{ij}e^{-(\xi _{N+i}+\bar{\xi}_{N+j})}+\frac{1}{%
p_{i}+p_{j}^{\ast }} & \frac{1}{p_{i}+p_{N+j}^{\ast }} \\
\frac{1}{p_{N+i}+p_{j}^{\ast }} & c_{N+i}^{\ast }\delta _{i,j}e^{-(\xi_i+%
\bar{\xi}_{j})}+\frac{1}{p_{N+i}+p_{N+j}^{\ast }}%
\end{array}%
\right\vert
\end{eqnarray}
On the other hand,
\begin{eqnarray}
\fl f(x,t) &=&\left\vert
\begin{array}{cc}
c_{i}\delta _{ij}e^{-(\xi_i+\bar{\xi}_{j})}+\frac{1}{p_{i}+p_{j}^{\ast }} &
\frac{1}{p_{i}+p_{N+j}^{\ast }} \\
\fl \frac{1}{p_{N+i}+p_{j}^{\ast }} & c_{N+i}\delta _{N+i,N+j}e^{-(\xi _{N+i}+%
\bar{\xi}_{N+j})}+\frac{1}{p_{N+i}+p_{N+j}^{\ast }}%
\end{array}%
\right\vert \nonumber \\
\fl &=&\left\vert
\begin{array}{cc}
c_{j}\delta _{ji}e^{-(\xi _{j}+\bar{\xi}_{i})}+\frac{1}{p_{j}+p_{i}^{\ast }}
& \frac{1}{p_{j}^{\ast }+p_{N+j}} \\
\fl \frac{1}{p_{N+i}^{\ast }+p_{j}} & c_{N+j}\delta _{N+j,N+i}e^{-(\xi _{N+j}+%
\bar{\xi}_{N+i})}+\frac{1}{p_{N+i}^{\ast }+p_{N+j}}%
\end{array}%
\right\vert\,,
\end{eqnarray}
thus
\begin{equation}
\fl
f^{\ast }(-x,t)=\left\vert
\begin{array}{cc}
c_{j}\delta _{ji}e^{-(\xi _{i}+\bar{\xi}_{j})^{\ast }(-x,t)}+\frac{1}{%
p_{i}+p_{j}^{\ast }} & \frac{1}{p_{i}+p_{N+j}^{\ast }} \\
\frac{1}{p_{N+i}+p_{j}^{\ast }} & c_{N+j}\delta _{N+j,N+i}e^{-(\xi _{N+i}+%
\bar{\xi}_{N+j})^{\ast }(-x,t)}+\frac{1}{p_{N+i}+p_{N+j}^{\ast }}%
\end{array}%
\right\vert\,.
\end{equation}
Obviously $f(x,t)=f^{\ast }(-x,t)$. Next, let us proceed to prove $g(x,t)=g^{\ast}(-x,t)$. Note $g(x,t)$
can be written as
\begin{eqnarray}
\fl g(x,t) &=&\left\vert
\begin{array}{cc}
c_{i}\delta _{ij}\left( -\frac{p_{i}}{p_{j}^{\ast }}\right) e^{-(\xi _{i}+%
\bar{\xi}_{j})}+\frac{1}{p_{i}+p_{j}^{\ast }} & \frac{1}{p_{i}+p_{N+j}^{\ast
}}\,, \nonumber \\
\fl \frac{1}{p_{N+i}+p_{j}^{\ast }} & c_{N+i}\delta _{N+i,N+j}\left( -\frac{%
p_{N+i}}{p_{N+j}^{\ast }}\right) e^{-(\xi _{N+i}+\bar{\xi}_{N+j})}+\frac{1}{%
p_{N+i}+p_{N+j}^{\ast }}%
\end{array}%
\right\vert, \nonumber \\
\fl &=&\left\vert
\begin{array}{cc}
c_{N+i}\delta _{N+i,N+j}\left( -\frac{p_{N+i}}{p_{N+j}^{\ast }}\right)
e^{-(\xi _{N+i}+\bar{\xi}_{N+j})}+\frac{1}{p_{N+i}+p_{N+j}^{\ast }} & \frac{1%
}{p_{N+i}+p_{j}^{\ast }}\,,\nonumber \\
\fl \frac{1}{p_{i}+p_{N+j}^{\ast }} & c_{i}\delta _{ij}\left( -\frac{p_{i}}{%
p_{j}^{\ast }}\right) e^{-(\xi _{i}+\bar{\xi}_{j})}+\frac{1}{%
p_{i}+p_{j}^{\ast }}%
\end{array}%
\right\vert, \nonumber \\
\fl &=&\left\vert
\begin{array}{cc}
-c_{i}^{\ast }\delta _{ij}\left( -\frac{p_{i}}{p_{j}^{\ast }}\right)
e^{-(\xi _{N+i}+\bar{\xi}_{N+j})}-\frac{1}{p_{i}+p_{j}^{\ast }} & -\frac{1}{%
p_{i}+p_{N+j}^{\ast }},\nonumber \\
\fl -\frac{1}{p_{N+i}+p_{j}^{\ast }} & -c_{N+i}^{\ast }\left( -\frac{p_{N+i}}{%
p_{N+j}^{\ast }}\right) \delta _{N+i,N+j}e^{-(\xi _{i}+\bar{\xi}_{j})}-\frac{%
1}{p_{N+i}+p_{N+j}^{\ast }}%
\end{array}%
\right\vert, \nonumber\\
\fl &=&\left\vert
\begin{array}{cc}
c_{i}^{\ast }\delta _{ij}\left( -\frac{p_{i}}{p_{j}^{\ast }}\right) e^{-(\xi
_{N+i}+\bar{\xi}_{N+j})}+\frac{1}{p_{i}+p_{j}^{\ast }} & \frac{1}{%
p_{i}+p_{N+j}^{\ast }} \\
\fl \frac{1}{p_{N+i}+p_{j}^{\ast }} & c_{N+i}^{\ast }\left( -\frac{p_{N+i}}{%
p_{N+j}^{\ast }}\right) \delta _{N+i,N+j}e^{-(\xi _{i}+\bar{\xi}_{j})}+\frac{%
1}{p_{N+i}+p_{N+j}^{\ast }}%
\end{array}%
\right\vert,
\end{eqnarray}
thus
\begin{equation}
\fl g^{\ast }(-x,t)=\left\vert
\begin{array}{cc}
c_{i}\delta _{ij}\left( -\frac{p_{i}^{\ast }}{p_{j}}\right) e^{-(\xi _{N+i}+%
\bar{\xi}_{N+j})^{\ast }(-x,t)}+\frac{1}{p_{i}^{\ast }+p_{j}} & \frac{1}{%
p_{i}^{\ast }+p_{N+j}} \\
\frac{1}{p_{N+i}^{\ast }+p_{j}} & c_{N+i}\left( -\frac{p_{N+i}^{\ast }}{%
p_{N+j}}\right) \delta _{ij}e^{-(\xi _{i}+\bar{\xi}_{j})^{\ast }(-x,t)}+%
\frac{1}{p_{N+i}^{\ast }+p_{N+j}}%
\end{array}%
\right\vert\,.
\end{equation}
On the other hand,
\begin{eqnarray}
\fl \bar{g}(x,t) &=&\left\vert
\begin{array}{cc}
c_{i}\delta _{ij}\left( -\frac{p_{i}}{p_{j}^{\ast }}\right) ^{-1}e^{-(\xi
_{i}+\bar{\xi}_{j})}+\frac{1}{p_{i}+p_{j}^{\ast }} & \frac{1}{%
p_{i}+p_{N+j}^{\ast }} \nonumber \\
\fl \frac{1}{p_{N+i}+p_{j}^{\ast }} & c_{N+i}\delta _{N+j,N+i}\left( -\frac{%
p_{N+i}}{p_{N+j}^{\ast }}\right) ^{-1}e^{-(\xi _{N+i}+\bar{\xi}_{N+j})}+%
\frac{1}{p_{N+i}+p_{N+j}^{\ast }}%
\end{array}%
\right\vert, \nonumber \\
\fl &=&\left\vert
\begin{array}{cc}
c_{j}\delta _{ji}\left( -\frac{p_{j}}{p_{i}^{\ast }}\right) ^{-1}e^{-(\xi
_{j}+\bar{\xi}_{i})}+\frac{1}{p_{j}+p_{i}^{\ast }} & \frac{1}{%
p_{N+j}+p_{i}^{\ast }} \nonumber \\
\fl \frac{1}{p_{j}+p_{N+i}^{\ast }} & c_{N+j}\delta _{N+j,N+i}\left( -\frac{%
p_{N+j}}{p_{N+i}^{\ast }}\right) ^{-1}e^{-(\xi _{N+j}+\bar{\xi}_{N+i})}+%
\frac{1}{p_{N+i}^{\ast }+p_{N+j}}%
\end{array}%
\right\vert\,.
\end{eqnarray}
Obviously $\bar{g}(x,t)=g^{\ast }(-x,t)$.
\end{proof}
{Based on above Lemma, two sets of bilinear equations (\ref{ndNLS_bilinear}) and (\ref{ndNLS_bilinear2}) coincide.}
Therefore, we have the following theorem regarding the general soliton solution.
\begin{theorem}
The nonlocal NLS equation admits the solution
\begin{equation}
\label{dNLS_sol1}
q(x,t)=\rho \frac{g(x,t)}{f(x,t)}e^{\mathrm{i}2\rho ^{2}t}\,,
\end{equation}
where
\begin{equation}
\label{dNLS_sol2}
 f(x,t)=\left\vert
c_{i}\delta _{ij}e^{-(\xi _{i}+\bar{\xi}_{j})}+\frac{1}{p_{i}+p_{j}^{\ast }}
\right\vert _{0\leq i,j\leq 2N}\,,
\end{equation}
\begin{equation}
\label{dNLS_sol3}
g(x,t)=\left\vert
c_{i}\delta _{ij}\left( -\frac{p_{i}}{p_{j}^{\ast }}\right) e^{-(\xi _{i}+%
\bar{\xi}_{j})}+\frac{1}{p_{i}+p_{j}^{\ast }}
\right\vert _{0\leq i,j\leq 2N}\,,
\end{equation}
with $c_{N+i}=-c_{i}^{\ast }$, $p_{N+i}=-p_{i}$, $\xi_{0,N+i}=\xi _{0i}$ for $i=1,\cdots,N$.
\end{theorem}
In what follows, we  give explicit form of the solution for $N=1$ and detailed analysis.
By taking $N=1$ in (\ref{dNLS_sol2}) and (\ref{dNLS_sol3}) and $p_1=\rho e^{-\mathrm{i} \theta}$
we get the tau functions
\begin{eqnarray}
 f &=&\left\vert
\begin{array}{cc}
c_{1}e^{-(\xi _{1}+\bar{\xi}_{1})}+\frac{1}{p_{1}+p_{1}^{\ast }} & \frac{1}{%
p_{1}+p_{2}^{\ast }} \\
 \frac{1}{p_{2}+p_{1}^{\ast }} & -c_{1}^{\ast }e^{-(\xi _{2}+\bar{\xi}_{2})}+%
\frac{1}{p_{2}+p_{2}^{\ast }}%
\end{array}%
\right\vert \nonumber\\
\fl &=&-\rho ^{-2}(|c_{1}|^{2}\rho ^{2}e^{4\rho ^{2}\sin 2\theta t}+\frac{%
1}{2}\rho c_{1}\sec \theta e^{-2\rho \cos \theta x+2\rho ^{2}\sin 2\theta t}
\nonumber \\
&&+\frac{1}{2}\rho c_{1}^{\ast }\sec \theta e^{2\rho \cos \theta x+2\rho
^{2}\sin 2\theta t}+\csc ^{2}2\theta)\,,
\end{eqnarray}

\begin{eqnarray}
 g &=&\left\vert
\begin{array}{cc}
c_{1}e^{-(\xi _{1}+\bar{\xi}_{1})}+\frac{1}{p_{1}+p_{1}^{\ast }}\left( -\frac{p_{1}}{p_{1}^{\ast }}%
\right)  & \frac{1}{p_{1}+p_{2}^{\ast }}\left( -%
\frac{p_{1}}{p_{2}^{\ast }}\right)  \\
 \frac{1}{p_{2}+p_{1}^{\ast }}\left( -\frac{p_{2}}{p_{1}^{\ast }}\right)
 & -c_{1}^{\ast }e^{-(\xi _{2}+\bar{\xi}_{2})}+\frac{1}{p_{2}+p_{2}^{\ast }}%
\left( -\frac{p_{2}}{p_{2}^{\ast }}\right) %
\end{array}%
\right\vert \nonumber \\
\fl &=&-e^{2\mathrm{i}\theta} \rho ^{-2}(|c_{1}|^{2}\rho ^{2}e^{4\rho ^{2}\sin 2\theta
t-2\mathrm{i}\theta}-\frac 12\rho c_{1}\sec \theta e^{-2\rho \cos \theta x+ 2\rho ^{2}\sin 2\theta t
} \nonumber \\
 &&-\frac 12\rho c_{1}^{\ast }\sec \theta e^{2\rho \cos \theta x+2\rho ^{2}\sin
2\theta t}+\csc ^{2}2\theta e^{2\mathrm{i}\theta})\,.
\end{eqnarray}
 {It can be easily shown that the resulting solution $q(x,t)=\rho \frac{g(x,t)}{f(x,t)}e^{\mathrm{i}2 \rho ^{2}t}$ corresponds to the solution (5.90) in \cite{AblowitzLuoMusslimani} with $\sigma=1$ and $\Delta \theta=0$}.
Next, we investigate the asymptotic behavior of above two-soliton solution. To
this end, we assume $0< \theta < \pi/2$ without loss of generality, and $\kappa=2\rho \cos \theta >0$. We define the right-moving soliton along the line
$\eta_1=x-2\rho \sin \theta t$ as soliton 1, and the left-moving soliton along the line $\eta_2=x+2\rho \sin \theta t$ as soliton 2.
For the above choice of parameters, we have
(i) $\eta_1 \approx 0$, $\eta_2 \rightarrow \mp \infty $ as $%
t\rightarrow \mp \infty $ for soliton 1 and (ii) $\eta_2 \approx 0$, $%
\eta_1 \rightarrow \pm \infty $ as $t\rightarrow \mp \infty $ for soliton
2. This leads to the following asymptotic forms for {the} two-soliton solution.
\newline
(i) Before collision ($t\rightarrow -\infty $)

Soliton 1 ($\eta _{1}\approx 0$, $\eta _{2}\rightarrow -\infty $):
\begin{eqnarray}
 q  &\rightarrow &   \rho \frac{\csc ^{2}\theta-\frac 12\rho c_{1}\sec \theta e^{-\kappa \eta_1 -2\mathrm{i}\theta} }
{\csc ^{2}\theta+\frac 12\rho c_{1} \sec \theta e^{-\kappa \eta_1}} e^{\mathrm{i}2\rho ^{2}t+4\mathrm{i}\theta}\,, \nonumber \\
&=&  \rho \frac{1- 2\rho |c_{1}|\sin^2 \theta \cos \theta e^{-\kappa \eta_1 -\mathrm{i}(2\theta-\phi)} }
{1+2\rho |c_{1}|\sin^2 \theta \cos \theta e^{-\kappa \eta_1+\mathrm{i} \phi}} e^{\mathrm{i}2\rho ^{2}t+4\mathrm{i}\theta}\,,
\label{soliton1_aybf}
\end{eqnarray}%
\begin{eqnarray}
&& |q|^2  \rightarrow   \rho^2 \frac{\cosh (\eta_1-\eta^{-}_{10})-\cos (2\theta-\phi)}
{\cosh (\eta_1-\eta^{-}_{10})+\cos \phi}\,.
\label{soliton1_aybf2}
\end{eqnarray}%

Soliton 2 ($\eta_{2} \approx 0$, $\eta_{1} \to \infty$):
\begin{eqnarray}
&& q  \rightarrow   \rho \frac{\csc ^{2}2\theta-\frac 12\rho c^{\ast}_{1}\sec \theta e^{-\kappa \eta_2 -2\mathrm{i}\theta} }
{\csc ^{2}2\theta+\frac 12\rho c^{\ast}_{1} \sec \theta e^{-\kappa \eta_2}} e^{\mathrm{i}2\rho ^{2}t+4\mathrm{i}\theta}\,, \nonumber \\
&=&  \rho \frac{1- 2\rho |c_{1}|\sin^2 \theta \cos \theta e^{-\kappa \eta_1 -\mathrm{i}(2\theta+\phi)} }
{1+2\rho |c_{1}|\sin^2 \theta \cos \theta e^{-\kappa \eta_1-\mathrm{i} \phi}} e^{\mathrm{i}2\rho ^{2}t+4\mathrm{i}\theta}\,,
\label{soliton2_aybf}
\end{eqnarray}%

\begin{eqnarray}
&& |q|^2  \rightarrow   \rho^2 \frac{\cosh (\eta_1-\eta^{-}_{10})-\cos (2\theta+\phi)}
{\cosh (\eta_1-\eta^{-}_{10})+\cos \phi}\,,
\label{soliton2_aybf2}
\end{eqnarray}%
where $e^{\eta^{-}_{10}}=2\rho |c_1| \sin^2 \theta \cos \theta$.

(ii) After collision ($t \to \infty$) \\
Soliton 1 ($\eta_{1} \approx 0$, $\eta_{2} \to \infty$):
\begin{eqnarray}
 q  &\rightarrow&   \rho \frac{|c_{1}|^{2}\rho ^{2}-\frac 12\rho c^{\ast}_{1}\sec \theta e^{\kappa \eta_1 +2\mathrm{i}\theta} }
{|c_{1}|^{2}\rho ^{2}+\frac 12\rho c^{\ast}_{1}\sec \theta e^{\kappa \eta_1 }} e^{\mathrm{i}2\rho ^{2}t}\,, \nonumber \\
&=& \rho \frac{2\rho |c_1| \cos \theta e^{-\kappa \eta_1 -\mathrm{i}(2\theta-\phi)}-1} {2\rho |c_1| \cos \theta e^{-\kappa \eta_1 +\mathrm{i}\phi}+1}
e^{\mathrm{i}2\rho ^{2}t+4\mathrm{i}\theta}
\label{soliton1_ayaf}
\end{eqnarray}%
\begin{eqnarray}
&& |q|^2  \rightarrow   \rho^2 \frac{\cosh (\eta_1-\eta^{+}_{10})-\cos (2\theta-\phi)}
{\cosh (\eta_1-\eta^{+}_{10})+\cos \phi}\,.
\label{soliton1_ayaf2}
\end{eqnarray}%
Soliton 2 ($\eta _{2}\approx 0$, $\eta _{1}\rightarrow -\infty $):
\begin{eqnarray}
 q  &\rightarrow&   \rho \frac{|c_{1}|^{2}\rho ^{2}-\frac 12\rho c_{1}\sec \theta e^{\kappa \eta_2 +2\mathrm{i}\theta} }
{|c_{1}|^{2}\rho ^{2}+\frac 12\rho c_1 \sec \theta e^{\kappa \eta_2 }} e^{\mathrm{i}2\rho ^{2}t}\,, \nonumber \\
&=& \rho \frac{2\rho |c_1| \cos \theta e^{-\kappa \eta_1 -\mathrm{i}(2\theta+\phi)}-1} {2\rho |c_1| \cos \theta e^{-\kappa \eta_1 -\mathrm{i}\phi}+1}
e^{\mathrm{i}2\rho ^{2}t+4\mathrm{i}\theta}
\label{soliton2_ayaf}
\end{eqnarray}%
\begin{eqnarray}
&& |q|^2  \rightarrow   \rho^2 \frac{\cosh (\eta_1-\eta^{+}_{10})-\cos (2\theta+\phi)}
{\cosh (\eta_1-\eta^{+}_{10})+\cos \phi}\,,
\label{soliton2_ayaf2}
\end{eqnarray}%
where $e^{\eta^{+}_{10}}=2\rho |c_1|  \cos \theta$.

\begin{figure}[tbh]
\centering
\includegraphics[height=60mm,width=80mm]{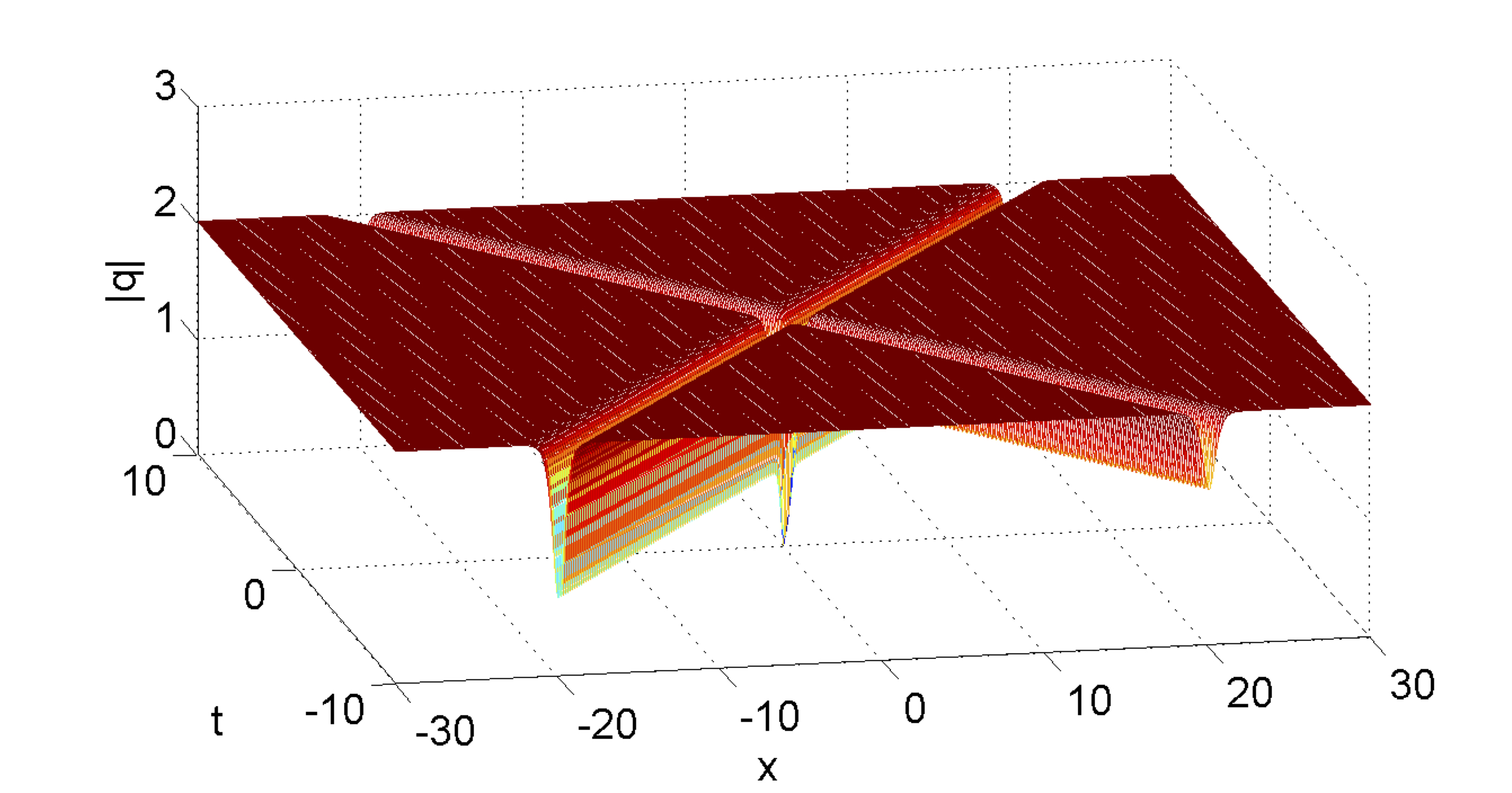}
\caption{Dark-dark soliton with $\rho=2.0$, $c_1=1.0+0.4i$ and $\theta=\pi/6$.}
\label{fig1}
\end{figure}

\begin{figure}[tbh]
\centering
\includegraphics[height=60mm,width=80mm]{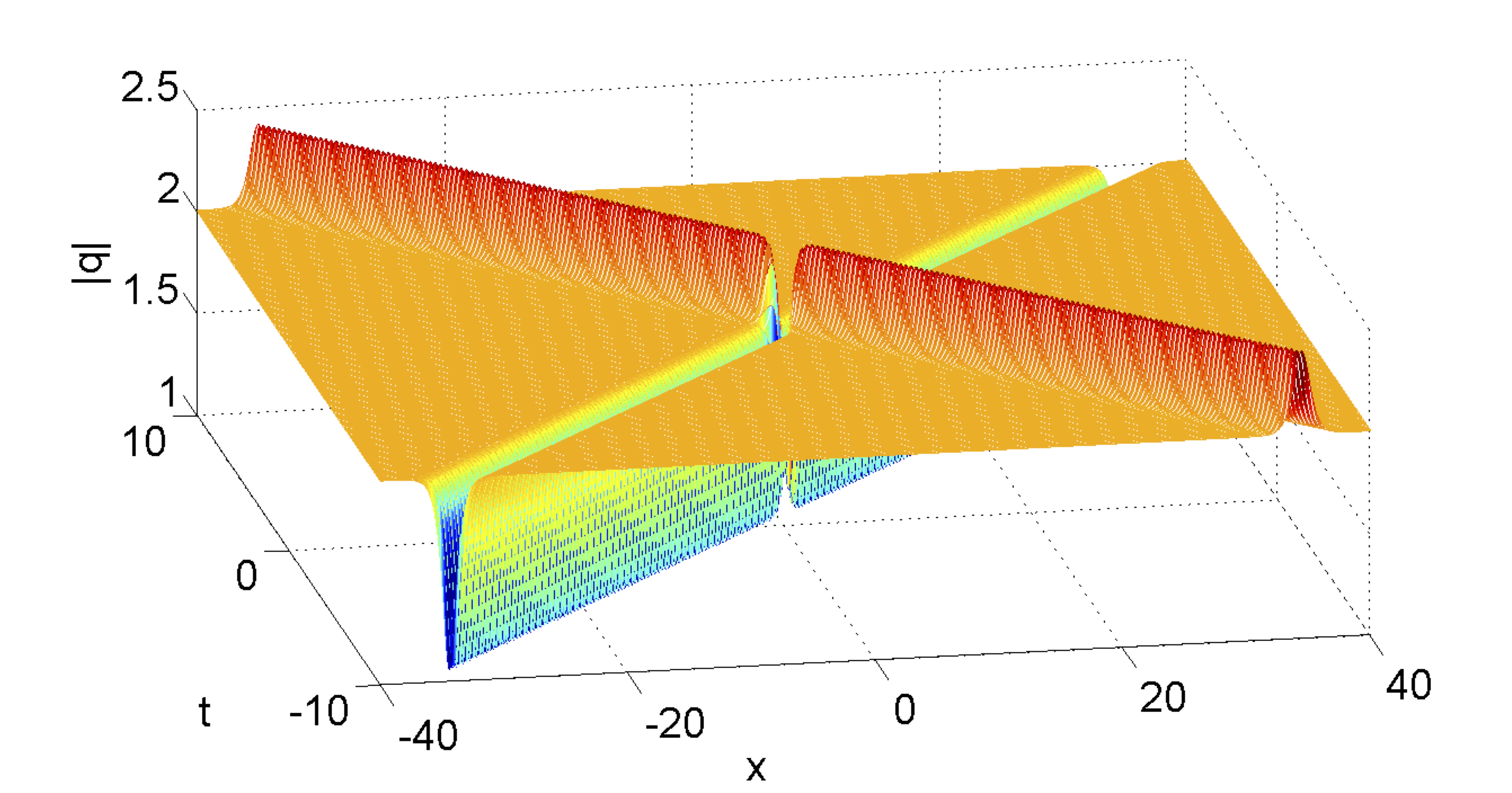}
\caption{Dark-antidark soliton with $\rho=2.0$, $c_1=0.4+i$ and $\theta=\pi/3$.}
\label{fig2}
\end{figure}

\begin{figure}[tbh]
\centering
\includegraphics[height=60mm,width=80mm]{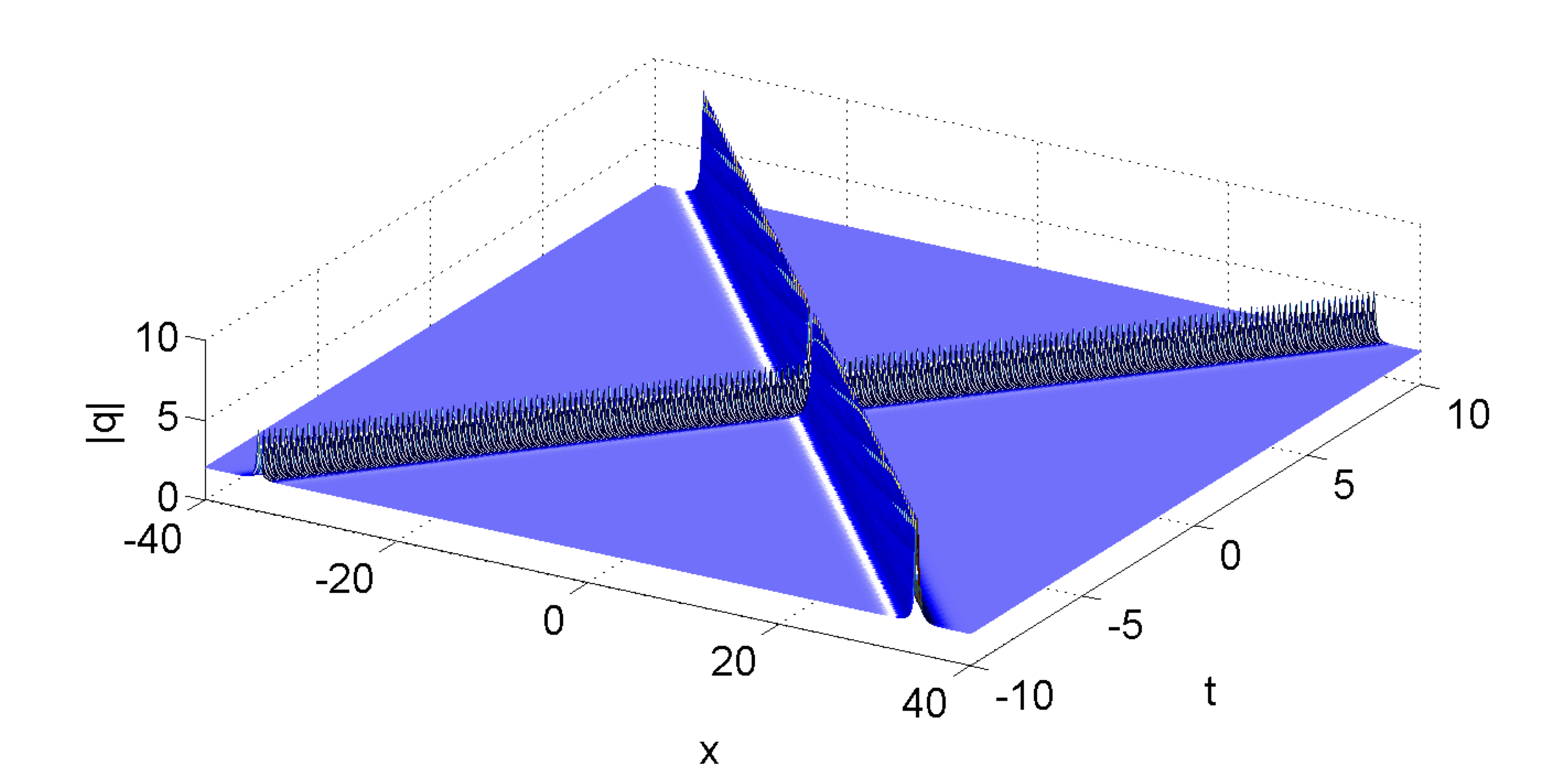}
\caption{Antidark-antidark soliton with $\rho=2.0$, $c_1=-1.0+0.3i$ and $\theta=\pi/3$.}
\label{fig3}
\end{figure}
\subsection{Soliton solution with nonzero boundary condition to the nonlocal NLS equation {with $\sigma =-1, \Delta \theta=\pi.$}
}
In this subsection, we consider the soliton solutions to the nonlocal NLS equation with $\sigma=-1$
\begin{equation}
\mathrm{i}q_{t}(x,t)=q_{xx}(x,t)+2q^{\ast }(-x,t)\,q^{2}(x,t)
\end{equation}
under boundary condition $q(x,t)\rightarrow \rho $ as $x\rightarrow \pm
\infty $. To this end, we introduce a variable transformation
\[
q(x,t)=\rho \frac{g(x,t)}{f(x,t)}e^{ 2\mathrm{i} \epsilon \rho ^{2}t},
\]%
with the condition $f(x,t)=\pm f^{\ast}(-x,t)$, then the nonlocal NLS equation is converted into the following bilinear
equations
 {
\begin{equation}
\label{bilinear_negative}
\left\{
\begin{array}{l}
\displaystyle(\mathrm{i}D_{t}-D_{x}^{2})g(x,t)\cdot f(x,t)=0, \\
\displaystyle(D_{x}^{2}-2\epsilon \rho ^{2})f(x,t)\cdot f(x,t)= \pm 2 \rho^{2}g(x,t)g^{\ast}(-x,t)\,,%
\end{array}%
\right.
\end{equation}
}
We recall the tau functions of single component KP hierarchy $\tau_{k}$ and the
bilinear equations (\ref{GramNZ_blinear}) the tau functions satisfy in previous section.
We take the simplest case $p_{1}=\bar{p}_{1}=\rho$, then the second bilinear equation in
(\ref{GramNZ_blinear}) reduces to the second equation in (\ref{Gram_blinear_dim}).
Furthermore, if we let $c_{1}=\frac{1}{2\rho}$, and define
 {
\begin{equation}
  f(x,t)=\tau_0=\frac{1}{2\rho }\left( 1+e^{\xi _{1}+\bar{\xi }_{1}}\right)\,,
\end{equation}
\begin{equation}
  g(x,t)=\tau_1=\frac{1}{2\rho }\left( 1-e^{\xi _{1}+\bar{\xi }_{1}}\right)\,,
\end{equation}
\begin{equation}
  \bar{g}(x,t)=\tau_{-1}=\frac{1}{2\rho }\left( 1-e^{\xi _{1}+\bar{\xi }_{1}} \right)\,,
\end{equation}
where $e^{\xi _{1}+\bar{\xi }_{1}}=e^{2 \rho x}$. It then follows
\begin{equation}
  f(x,t)=e^{(\xi _{1}+\bar{\xi }_{1})} f^{\ast }(-x,t)\,, \quad \bar{g}(x,t)=-e^{(\xi _{1}+\bar{\xi }_{1})} g^{\ast}(-x,t)\,.
\end{equation}
Furthermore, if we let $x_1=x$, $x_2=- \mathrm{i} t$ and $\epsilon=1$, the bilinear equation (\ref{Gram_blinear_dim}) reduces to the bilinear equation (\ref{bilinear_negative}).
}  Under this case, if we define
\begin{equation}
  q(x,t)=\frac{g(x,t)}{f(x,t)}e^{\mathrm{i}2\rho ^{2}t}\,, \quad r(x,t)=\frac{\bar{g}(x,t)}{f(x,t)}e^{-\mathrm{i}2\rho ^{2}t}\,,
\end{equation}
it is obvious that
\begin{equation}
  r(x,t)=-q^{\ast }(-x,t),
\end{equation}
which leads to the 1-soliton solution
 {
\begin{equation}
\label{Sig-1Detpi}
  q(x,t)=\rho \frac{ 1-e^{\xi _{1}+\bar{\xi }_{1}}}{ 1+e^{\xi _{1}+\bar{\xi }_{1}}} e^{2 \mathrm{i} \rho ^{2}t}
  =\rho e^{2 \mathrm{i} \rho ^{2}t} \tanh (\rho x) \,.
\end{equation}
This solution coincides with the solution (3.200) and (3.201) found in \cite{AblowitzLuoMusslimani} {for the case: $\sigma =-1, \Delta \theta=\pi$.}
}
\subsection{Soliton solution with nonzero boundary condition to the nonlocal NLS equation {with $\sigma =-1, \Delta \theta=0.$}
}
Again for the tau functions of single component KP hierarchy $\tau_{k}$ and the
bilinear equations the tau functions satisfy in previous section. If we impose the condition
$p_{j}\bar{p}_{j}=-\rho ^{2}$,\ $j=1,\cdots ,N$\,,
then similarly, we have
\begin{equation}
\partial_{x_{1}}\tau _{k}=-\rho^{2}\partial _{x_{-1}}\tau _{k}\,.
\end{equation}
Furthermore, if we assume $x_{1}=x$, $x_{2}=-\mathrm{i}t$,
then the bilinear equation (\ref{GramNZ_blinear}) reduces to
{\begin{equation}
\label{bilinear_negative2}
\left\{
\begin{array}{l}
\displaystyle(\mathrm{i}D_{t}-D_{x}^{2}) \tau _{1}(x,t) \cdot \tau _{0}(x,t)=0, \\
\displaystyle(D_{x}^{2}+2\rho ^{2})\tau _{0}(x,t)  \cdot \tau _{0}(x,t)=  2 \rho ^{2}\tau _{1}(x,t)  \tau _{-1}(x,t)\,.%
\end{array}%
\right.
\end{equation}
Consider a $2 \times 2$ matrix for $\tau _{0}(x,t)$, $\tau _{1}(x,t)$ and $\tau _{-1}(x,t)$, let $p_{1}=v_{1}$, $\bar{p}_{1}=-\rho ^{2}/v_{1}$
where $v_{1}>\rho $ is real, $p_{2}=\bar{p}_{1}$, $\bar{p}_{2}=p_{1}$, $c_{2}=-c_{1}$ and
$$
c_{1}^{2}=\frac{1}{%
(p_{1}+\bar{p}_{1})^2}-\frac{1}{4p_{1}\bar{p}_{1}}
$$. \
Since
\[
(\xi _{1}+\bar{\xi }_{1})(x,t)=(v_{1}-\rho ^{2}/v_{1})x+(v_{1}^{2}-\rho
^{4}/v_{1}^{2})\mathrm{i}t\,, \
\]
\[
(\xi _{2}+\bar{\xi}_{2})(x,t)=(v_{1}-\rho ^{2}/v_{1})x-(v_{1}^{2}-\rho
^{4}/v_{1}^{2})\mathrm{i}t\,,
\]
\begin{eqnarray}
\fl \tau _{0}(x,t) &=&\left\vert
\begin{array}{cc}
c_{1}+\frac{1}{p_{1}+\bar{p}_{1}}e^{\xi _{1}+\bar{\xi}_{1}} & \frac{1}{p_{1}+%
\bar{p}_{2}}e^{\xi _{1}+\bar{\xi}_{2}} \\
 \frac{1}{p_{2}+\bar{p}_{1}}e^{\xi _{2}+\bar{\xi}_{1}} & c_{2}+\frac{1}{p_{2}+%
\bar{p}_{2}}e^{\xi _{2}+\bar{\xi}_{2}}%
\end{array}%
\right\vert \nonumber \\
\fl &=&-c_{1}{}^{2}+\frac{c_{1}}{p_{2}+\bar{p}_{2}}e^{\xi _{2}+\bar{\xi}_{2}}-%
\frac{c_{1}}{p_{1}+\bar{p}_{1}}e^{\xi _{1}+\bar{\xi}_{1}}+c_{1}{}^{2}e^{\xi
_{1}+\bar{\xi}_{1}+\xi _{2}+\bar{\xi}_{2}}\, \nonumber \\
\fl &=&2c_{1} e^{(v_1-\rho^2/v_1)x}\left\{c_{1} \sinh \left[(v_1-\rho^2/v_1)x \right] -\frac{\mathrm{i}}{p_{1}+\bar{p}_{1}} \sin \left[(v_1^2-\rho^4/v_1^2)t \right]\right\}\,.
\end{eqnarray}
Therefore
\begin{eqnarray}
\fl  \tau _{0}^{\ast }(-x,t) &=&-e^{-2(v_1-\rho^2/v_1)x} \tau _{0}(x,t)\,.
\end{eqnarray}
Moreover
\begin{eqnarray}
\fl  \tau _{1}(x,t) &=& \left\vert
\begin{array}{cc}
c_{1}+\frac{1}{p_{1}+\bar{p}_{1}}\left( -\frac{p_{1}}{\bar{p}_{1}}\right)
e^{\xi _{1}+\bar{\xi}_{1}} & \frac{1}{p_{1}+\bar{p}_{2}}\left( -\frac{p_{1}}{%
\bar{p}_{2}}\right) e^{\xi _{1}+\bar{\xi}_{2}} \\
 \frac{1}{p_{2}+\bar{p}_{1}}\left( -\frac{p_{2}}{\bar{p}_{1}}\right) e^{\xi
_{2}+\bar{\xi}_{1}} & c_{2}+\frac{1}{p_{2}+\bar{p}_{2}}\left( -\frac{p_{2}}{%
\bar{p}_{2}}\right) e^{\xi _{2}+\bar{\xi}_{2}}%
\end{array}%
\right\vert \nonumber \\
\fl &=&-c_{1}^{2}{}+\frac{c_{1}}{p_{2}+\bar{p}_{2}}\left( -\frac{p_{2}}{%
\bar{p}_{2}}\right) e^{\xi _{2}+\bar{\xi}_{2}}-\frac{c_{1}}{p_{1}+\bar{p}_{1}%
}\left( -\frac{p_{1}}{\bar{p}_{1}}\right) e^{\xi _{1}+\bar{\xi}%
_{1}}+c_{1}^{2}{}e^{\xi _{1}+\bar{\xi}_{1}+\xi _{2}+\bar{\xi}_{2}}\, \nonumber \\
\fl &=&c_{1} e^{(v_1-\frac{\rho^2}{v_1})x}\left\{2 c_{1} \sinh \left[(v_1-\frac{\rho^2}{v_1})x \right] +\frac{1}{p_{1}+\bar{p}_{1}} \left[ \frac{p_{1}}{\bar{p}_{1}}e^{(v_1^2-\frac{\rho^4}{v^2_1})\mathrm{i}t}-%
\frac{\bar{p}_{1}}{p_{1}} e^{-(v_1^2-\frac{\rho^4}{v^2_1})\mathrm{i}t} \right]\right\}\,. \nonumber
\end{eqnarray}
\begin{eqnarray}
\fl \tau _{-1}(x,t) &=&\left\vert
\begin{array}{cc}
c_{1}+\frac{1}{p_{1}+\bar{p}_{1}}\left( -\frac{\bar{p}_{1}}{p_{1}}\right)
e^{\xi _{1}+\bar{\xi}_{1}} & \frac{1}{p_{1}+\bar{p}_{2}}\left( -\frac{\bar{p}%
_{2}}{p_{1}}\right) e^{\xi _{1}+\bar{\xi}_{2}} \\
 \frac{1}{p_{2}+\bar{p}_{1}}\left( -\frac{\bar{p}_{1}}{p_{2}}\right) e^{\xi
_{2}+\bar{\xi}_{1}} & c_{2}+\frac{1}{p_{2}+\bar{p}_{2}}\left( -\frac{\bar{p}%
_{2}}{p_{2}}\right) e^{\xi _{2}+\bar{\xi}_{2}}%
\end{array}%
\right\vert \nonumber \\
\fl &=& -c_{1}^{2}{}+\frac{c_{1}}{p_{2}+\bar{p}_{2}}\left( -\frac{\bar{p}%
_{2}}{p_{2}}\right) e^{\xi _{2}+\bar{\xi}_{2}}-\frac{c_{1}}{p_{1}+\bar{p}_{1}%
}\left( -\frac{\bar{p}_{1}}{p_{1}}\right) e^{\xi _{1}+\bar{\xi}%
_{1}}+c_{1}^{2}{}e^{\xi _{1}+\bar{\xi}_{1}+\xi _{2}+\bar{\xi}_{2}}\,\nonumber \\
\fl &=&c_{1} e^{(v_1-\frac{\rho^2}{v_1})x}\left\{2 c_{1} \sinh \left[(v_1-\frac{\rho^2}{v_1})x \right] +\frac{1}{p_{1}+\bar{p}_{1}} \left[ \frac{\bar{p}_{1}}{p_{1}} e^{(v_1^2-\frac{\rho^4}{v^2_1})\mathrm{i}t}-%
 \frac{p_{1}}{\bar{p}_{1}} e^{-(v_1^2-\frac{\rho^4}{v^2_1})\mathrm{i}t} \right]\right\}\,. \nonumber
\end{eqnarray}
Therefore
\begin{eqnarray}
\fl  \tau _{1}^{\ast }(-x,t) &=&-e^{-2(v_1-\rho^2/v_1)x} \tau _{-1}(x,t)\,.
\end{eqnarray}
Consequently, if we define
\begin{equation}
\tau _{0}(x,t)=Cf(x,t),\
\tau _{1}(x,t)=Cg(x,t),\
\tau _{-1}(x,t)=C\bar{g}(x,t),
\end{equation}
with $C=e^{(v_1-\rho^2/v_1)x}$, then we can easily show that
\begin{equation}
f(x,t)=-f^{\ast }(-x,t), \quad \bar{g}(x,t)=-g^{\ast }(-x,t),
\end{equation}
which reduces the bilinear equation (\ref{bilinear_negative2}) to
\begin{equation}
\label{bilinear_negative3}
\left\{
\begin{array}{l}
\displaystyle(\mathrm{i}D_{t}-D_{x}^{2}) g(x,t)\cdot f(x,t)=0, \\
\displaystyle(D_{x}^{2}+2\rho ^{2})f(x,t) \cdot f(x,t)=2\rho ^{2}{g}(x,t) g^*(-x,t)\,.%
\end{array}%
\right.
\end{equation}
Above bilinear equation coincides with the bilinear equation (\ref{bilinear_negative}) with $\epsilon=-1$ under the case
$f(x,t)=-f^*(-x,t)$. Moreover, for
\begin{equation}
  q(x,t)=\frac{g(x,t)}{f(x,t)}e^{-\mathrm{i}2\rho ^{2}t}\,, \quad r(x,t)=-\frac{\bar{g}(x,t)}{f(x,t)}e^{-\mathrm{i}2\rho ^{2}t}\,,
\end{equation}
it is obvious that
\begin{equation}
  r(x,t)=-q^{\ast }(-x,t),
\end{equation}
which leads to the 1-soliton solution
\begin{equation}
\label{Sig01Det0}
  q(x,t)=\rho \frac{ (v_1^2+\rho^2) \sinh \left[(v_1-\frac{\rho^2}{v_1})x \right] + \frac{2\rho^3}{v_1}e^{-(v_1^2-\frac{\rho^4}{v^2_1})\mathrm{i}t}-%
\frac{2v_1^3}{\rho} e^{(v_1^2-\frac{\rho^4}{v^2_1})\mathrm{i}t} }
  {(v_1^2+\rho^2) \sinh \left[(v_1-\rho^2/v_1)x \right] -4 \mathrm{i} \rho v_1 \sin \left[(v_1^2-\rho^4/v_1^2)t \right]}  e^{-2 \mathrm{i} \rho ^{2}t} \,.
\end{equation}
This solution coincides with the solution (4.130) in \cite{AblowitzLuoMusslimani} {for the case: $\sigma =-1, \Delta \theta=0$}.
}
\section{Summary and concluding remarks}
In this paper, via the combination of Hirota's bilinear method and the KP hierarchy reduction method, we have constructed various general soliton solutions in the form of determinants for {the} nonlocal NLS equation with both zero and nonzero boundary conditions. {The solutions obtained here} recover the solutions found in the nonlocal NLS equation so far and reformulate them in compact form by using tau functions of the KP hierarchy.

A natural extension of the present work will be the construction of the general soliton solutions to the coupled nonlocal NLS equation
\begin{equation}
\left\{
\begin{array}{l}
\displaystyle\mathrm{i}q_{1,t}(x,t)=q_{1,xx}(x,t)-2(\sigma_1 q_1^{2}(x,t)+\sigma_2 q_2^{2}(x,t))q^{\ast}_1(-x,t)\,, \\
\displaystyle \mathrm{i}q_{2,t}(x,t)=q_{2,xx}(x,t)-2(\sigma_1 q_1^{2}(x,t)+\sigma_2 q_2^{2}(x,t))q^{\ast}_2(-x,t)\,,\,.%
\end{array}%
\right.  \label{nlCNLS}
\end{equation}%
It is expected that the soliton solution to {the} above coupled nonlocal NLS equation is {even} richer and more complicated than the single component nonlocal NLS equation. We will explore this interesting topic and report the results elsewhere in the near future.
\section*{Acknowledgements}
MJA was partially supported by NSF under Grant No. {DMS-1712793}. BF was partially supported by NSF Grant under Grant No. (DMS-1715991) and the COS Research Enhancement Seed Grants Program at UTRGV.

\end{document}